\documentclass[12pt,letterpaper]{article}

\usepackage{amsmath,amsthm}
\usepackage{amsfonts}
\usepackage{algorithm,algorithmic}
\usepackage{amssymb,bbm}
\usepackage{latexsym}
\usepackage{graphicx}
\usepackage{color}
\usepackage{url}
\usepackage{xspace} 
\usepackage{balance}
\usepackage{authblk}
\definecolor{ForestGreen}{rgb}{0.1333,0.5451,0.1333}
\usepackage[letterpaper,
            colorlinks,linkcolor=ForestGreen,citecolor=ForestGreen,
            backref,
            bookmarks,bookmarksopen,bookmarksnumbered]
           {hyperref}

\newtheorem*{proof*}{\smallskip\noindent{\textbf{Proof}}} 
\newtheorem{theorem}{Theorem}
\newtheorem{problem}{Problem}
\newtheorem{claim}{Claim}
\newtheorem{lemma}{Lemma}
\newtheorem{corollary}{Corollary}
\newtheorem{definition}{Definition}

\newcommand{\hide}[1]{}

\newcommand{\field}[1]{\mathbb{#1}} 

\newcommand{\Mean}[1]{\ensuremath{{\mathbb E}\left[{#1}\right]}}

\newcommand{\NPhard}{{\ensuremath{\mathbf{NP}}-hard}\xspace}
\newcommand{\NPcomplete}{{\ensuremath{\mathbf{NP}}-complete}\xspace}

\newcommand{\mrc}{\textsc{MapReduce}\xspace}
\newcommand{\mTDS}{{\textsf{Triangle Density}}\xspace}
\newcommand{\TDS}{{\textsf{triangle-densest subgraph}}\xspace}
\newcommand{\KCDS}{{\textsf{$k$-clique-densest subgraph}}\xspace}
\newcommand{\fourCDS}{{\textsf{$4$-clique-densest subgraph}}\xspace}
\newcommand{\DS}{{\textsf{densest subgraph}}\xspace}

\newcommand{\DSP}{{\sc{DS-Problem}}\xspace}
\newcommand{\TDSP}{{\sc{TDS-Problem}}\xspace}
\newcommand{\KCDSP}{{\sc{k-Clique-DS-Problem}}\xspace}
\newcommand{\CTDSP}{{\sc{Constrainted-TDS-Problem}}\xspace}

\newcommand{\spara}[1]{\smallskip\noindent{\bf #1}}

\newcommand{\squishlist}{
 \begin{list}{$\bullet$}
  {  \setlength{\itemsep}{0pt}
     \setlength{\parsep}{3pt}
     \setlength{\topsep}{3pt}
     \setlength{\partopsep}{0pt}
     \setlength{\leftmargin}{2em}
     \setlength{\labelwidth}{1.5em}
     \setlength{\labelsep}{0.5em}
} }
\newcommand{\squishlisttight}{
 \begin{list}{$\bullet$}
  { \setlength{\itemsep}{0pt}
    \setlength{\parsep}{0pt}
    \setlength{\topsep}{0pt}
    \setlength{\partopsep}{0pt}
    \setlength{\leftmargin}{2em}
    \setlength{\labelwidth}{1.5em}
    \setlength{\labelsep}{0.5em}
} }

\newcommand{\squishdesc}{
 \begin{list}{}
  {  \setlength{\itemsep}{0pt}
     \setlength{\parsep}{3pt}
     \setlength{\topsep}{3pt}
     \setlength{\partopsep}{0pt}
     \setlength{\leftmargin}{1em}
     \setlength{\labelwidth}{1.5em}
     \setlength{\labelsep}{0.5em}
} }

\newcommand{\squishend}{
  \end{list}
}

\begin{document}

\title{A Novel Approach to Finding Near-Cliques:\\ The Triangle-Densest Subgraph Problem}

\author{
  Charalampos E. Tsourakakis\\
  ICERM, Brown University\\
  \texttt{charalampos\_tsourakakis@brown.edu} 
}

\maketitle \sloppy


\begin{abstract}  
Many graph mining applications rely on detecting subgraphs which are near-cliques.
There exists a dichotomy between the results in the existing work related to this problem: 
on the one hand the densest subgraph problem (\DSP) which maximizes the average 
degree over all subgraphs is solvable in polynomial time but 
for many networks, fails to find subgraphs which are near-cliques. 
On the other hand, formulations that are geared towards finding near-cliques
are \NPhard  and frequently inapproximable due to connections with the \textsf{Maximum Clique} problem. 

In this work, we propose a formulation which combines 
the best of both worlds: it is solvable in polynomial time 
and finds near-cliques when the \DSP fails. 
Surprisingly, our formulation is a simple variation of the \DSP. 
Specifically, we define the triangle densest subgraph problem (\TDSP): 
given $G(V,E)$, find a subset of vertices $S^*$ such that $\tau(S^*)=\max_{S \subseteq V}  \frac{t(S)}{|S|}$, 
where $t(S)$ is the number of triangles induced by the set $S$.  
We provide various exact and approximation algorithms which the solve \TDSP efficiently.
Furthermore, we show how our algorithms adapt to the 
more general problem of maximizing the $k$-clique average density. Finally, 
we provide empirical evidence that the \TDSP 
should be used whenever the output of \DSP fails to output a near-clique. 
\end{abstract}

\section{Introduction}
\label{sec:intro}
A wide variety of graph mining applications relies on extracting dense subgraphs
from large graphs. A list of some important such applications follows.

(1) Bader and Hogue observe that protein complexes, namely groups of 
proteins co-operating to achieve various biological functions, correspond to dense subgraphs 
in protein-protein interaction networks \cite{bader2003automated}. 
This observation is the cornerstone for several 
research projects which aim to identify such complexes, c.f.
\cite{brun2004clustering,pereira2004detection,prvzulj2004functional}.

(2) Sharan and Shamir notice that finding tight co-expression clusters
in microarray data can be reduced to finding dense co-expression subgraphs \cite{sharan2000click}. 
Hu et al.  capitalize on this observation to mine dense subgraphs across
a family of networks \cite{hu2005mining}. 

(3) Fratkin et al. show an approach to finding regulatory motifs in DNA based on 
finding dense subgraphs \cite{fratkin2006motifcut}. 

(4) Iasemidis et al. rely on dense subgraph extraction to study epilepsy \cite{iasemidis2001quadratic}. 

(5) Buehrer  and Chellapilla show how to compress Web graphs using 
as their main primitive the detection of dense subgraphs \cite{buehrer2008scalable}. 

(6) Gibson et al. observe that an algorithm which extracts dense subgraphs 
can be used to detect link spam in Web graphs \cite{gibson}.

(7) Dense subgraphs are used for finding stories and events in micro-blogging streams~\cite{Angel}.

(8) Alvarez-Hamelin et al.  rely on dense subgraphs to provide a better
understanding of the Internet topology \cite{alvarez2006k}.

(9) In the financial domain, extracting dense subgraphs has been applied to, among others,
predicting the behavior of financial instruments~\cite{BBP},
and finding price value motifs~\cite{Du}. 


Among the various formulations for finding dense subgraphs, 
the densest subgraph problem (\DSP) stands out 
for the facts that is solvable in polynomial time \cite{Goldberg84}
and  $\frac{1}{2}$-approximable in linear time \cite{AHI02,Char00,khuller}. 
To state the \DSP we introduce the necessary notation first. 
In this work we will focus on simple unweighted, undirected graphs. 
Given a graph $G=(V,E)$ and a subset of vertices 
$S\subseteq V$,  let $G(S)=(S,E(S))$ be the subgraph induced by $S$, and
let $e(S)=|E(S)|$ be the size of $E(S)$. 
Also, the \emph{edge density} of the set $S$ is defined 
as $f_e(S) = e(S) / {|S| \choose 2}$.
Notice that finding a subgraph which maximizes $f_e(S)$ is trivial.
Since $0 \leq f_e(S) \leq 1$ for any $S\subseteq V$, a single edge
achieves the maximum possible edge density. Therefore, 
the direct maximization of $f_e$ is not a meaningful problem.
The \DSP maximizes the ratio $\frac{e(S)}{|S|}$
over all subgraphs $S \subseteq V$. 
Notice that this is equivalent to maximizing the average degree. 
The \DSP is a powerful primitive for many graph applications including 
social piggybacking \cite{Gionis:2013:PSN:2536336.2536342} 
reachability and distance query indexing \cite{CohenHKZ02,JinXRF09}.
However, for many applications, including most of the listed applications, 
the goal is to find subgraphs which are near-cliques. 
Since the \DSP fails to find such subgraphs frequently by tending to favor
large subgraphs with not very large edge density $f_e$ 
other formulations have been proposed, see Section~\ref{sec:related}.
Unfortunately, these formulations are \NPhard and also inapproximable
due the connections with the \textsf{Maximum Clique} problem \cite{hastad}. 

The goal of this work is to propose a tractable formulation which 
extracts near-cliques when the \DSP fails.
\hide{ fits better in the related work
Unfortunately, the recent optimal quasi-clique problem 
proposed in \cite{tsourakakis2013denser} 
succeeds in finding near-cliques but is \NPhard and 
inapproximable \cite{mythesis}.}

\subsection{Contributions}

The main contribution of this work is the following:
we propose a novel objective which  attacks efficiently  
the problem of extracting near-cliques from large graphs,
an important problem for many applications, and is tractable.
Specifically, our contributions are summarized as follows.

\spara{New objective.} We introduce the {\em  average triangle density} 
 as  a novel objective for finding dense subgraphs. 
 We refer to the problem of maximizing the average triangle 
 density as the triangle-densest subgraph problem (\TDSP). 
 
\spara{Exact algorithms.} We develop three exact algorithms for the \TDSP. The 
 algorithm which achieves the best running time is based
 on maximum flow computations. It is worth outlining that 
 Goldberg's algorithm for the \DSP \cite{Goldberg84} 
 does not generalize to the \TDSP. 
 For this purpose, 
 we develop a novel approach that subsumes the \DSP
 and solves the \TDSP. Furthermore, our approach can solve
 a generalization of the \DSP and \TDSP that we introduce:
 maximize the   average {\it $k$-clique} density for any $k$ constant. 
 
\spara{Approximation algorithm.} We propose a $\frac{1}{3}$-approximation algorithm 
 for the \TDSP which runs asymptotically faster than any 
 of the exact algorithms. 

\spara{\mrc implementation.} We propose a $\frac{1}{3+3\epsilon}$-approximation
 algorithm for any $\epsilon>0$ which can be implemented efficiently in \mrc.
 The algorithm requires $O(\log(n)/\epsilon)$ rounds and is \mrc-efficient \cite{karloff2010model} due to the 
 existence of  efficient \mrc triangle counting algorithms \cite{suri2011counting}.
 
 \spara{Experimental evaluation.}  It is clear that in general the \DSP and the \TDSP can result
  in very different outputs. For instance, consider a graph 
  which is the union of a triangle and a large complete bipartite clique.
  The \DSP problem is optimized via the bipartite clique, the \TDSP via the triangle.
  Based on experiments the two objectives behave differently on 
  real-world networks as well.  
  For all datasets we have experimented with, we observe that the \TDSP consistently
  succeeds in extracting near-cliques. 
  For instance, in the \textsf{Football} network (see Table~\ref{tab:datasets}
  for a description of the dataset) the \DSP returns the whole graph 
  as the densest subgraph, with $f_e=0.094$ whereas the \TDSP returns 
  a subgraph on 18 vertices with $f_e=0.48$. 
  
  Therefore, the \TDSP 
  should be considered as an alternative to the \DSP 
  when the latter fails  to output near-cliques. 
  Also, we perform numerous experiments on real datasets which show that 
  the performance of the $\frac{1}{3}$-approximation algorithm 
  is close to the optimal performance. 
  
  \spara{Graph mining application.} We propose a modified version of the \TDSP, 
  the constrained triangle densest subgraph problem (\CTDSP), which aims to maximize the triangle 
  density subject to the constraint that the output should contain a prespecified 
  set of vertices $Q$.  We show how to solve exactly the \TDSP. 
  This variation is useful in various data-mining and bioinformatics tasks, see \cite{tsourakakis2013denser}.

The paper is organized as follows: Section~\ref{sec:related} presents related work.
Section~\ref{sec:problemdfn} defines and motivates the \TDSP. 
Section~\ref{sec:proposed} presents our theoretical contributions.
Section~\ref{sec:experiments} presents experimental findings on real-world networks. 
Section~\ref{sec:application} presents the \CTDSP.
Finally, Section~\ref{sec:conclusion} concludes the paper.

\section{Related Work}
\label{sec:related}
In Sections~\ref{subsec:finddense} and~\ref{subsec:trianglecounting}
we review related work to finding dense subgraphs and counting triangles respectively. 

\subsection{Finding Dense Subgraphs}
\label{subsec:finddense} 

\spara{Clique.} 
A clique is a set of vertices $S$ such that every two vertices in the subset are connected by an edge.
The \textsf{Clique} problem, i.e., finding whether there is a clique of a given size in a graph is \NPcomplete.
A maximum clique of a graph $G$ is a clique of maximum possible size
and its size is called the graph's clique number. Finding the clique number is \NPcomplete \cite{karp}. 
Furthermore, H$\mathring{a}$stad proved \cite{hastad} that  unless P = NP
there can be no polynomial time algorithm that approximates the 
maximum clique to within a factor better than $O(n^{1-\epsilon})$, for any $\epsilon > 0$.
When the max clique problem is parameterized by the order of the clique 
it is W[1]-hard \cite{downey}. 
Feige \cite{feige} proposed a polynomial time algorithm
that finds a clique of size $O\big( (\frac{\log{n}}{\log\log{n}})^2 \big)$
whenever the graph has a clique of size $O(\frac{n}{\log{n}^b})$ for any constant $b$. 
This algorithm leads to an algorithm that approximates the max clique 
within a factor of $O\big( n \frac{(\log\log{n})^2}{\log{n}^3}\big)$.
A maximal clique is a clique that is not a subset of a larger clique. 
A maximum clique is therefore always maximal, but the converse does not hold.
The Bron-Kerbosch algorithm \cite{BronKerbosch} is an exponential time method 
for finding all maximal cliques in a graph. A near optimal time algorithm 
for sparse graphs was introduced in \cite{Eppstein}.

\spara{Densest Subgraph.}  In the densest subgraph problem we are given a graph $G$ and we wish to 
find the set $S \subseteq  V$ which maximizes the average degree \cite{Goldberg84, Kannan}. 
The densest subgraph can be identified in polynomial time by solving a maximum flow problem
\cite{GGT89,Goldberg84}.  
Charikar \cite{Char00} proved that the greedy algorithm proposed by Asashiro et al.
\cite{AITT00} produces a $\frac{1}{2}$-approximation of the densest subgraph in linear time. 
Both algorithms are efficient in terms of running times and scale to large networks. 
In the case of directed graphs, the densest subgraph problem is solved in polynomial 
time as well \cite{Char00}. Khuller and Saha \cite{khuller} provide a linear time $\frac{1}{2}$-approximation
algorithm  for the case of directed graphs among other contributions. 
We notice that there is no size restriction of the output, i.e., $|S|$ could be arbitrarily large. When 
restrictions on the size of $S$ are imposed the problem becomes \NPhard. 
Specifically, the {\it DkS} problem, namely find the densest subgraph on $k$ vertices, is \NPhard \cite{AHI02}.
For general $k$, Feige, Kortsarz and Peleg \cite{FPK01} provide an approximation guarantee 
of $O(n^{\alpha})$ where $\alpha <1/3$.
Currently, the best approximation guarantee is $O(n^{1/4+\epsilon})$ for any $\epsilon>0$ 
due to Bhaskara et al. \cite{bhaskara2010detecting}.
The greedy algorithm of Asahiro et al.~\cite{AITT00} results in the approximation ratio $O(n/k)$. 
Therefore, when $k=\Omega(n)$ Asashiro et al. gave a constant factor approximation algorithm \cite{AITT00}.
It is worth mentioning that algorithms based on semidefinite programming have produced 
better approximation ratios for certain values of $k$ \cite{FL01}.
From the perspective of (in)approximability, Khot \cite{Khot06} proved that that there does
not exist any PTAS for the {\it DkS} problem under a reasonable complexity assumption.
Arora, Karger, and Karpinski~\cite{AKK95} gave a PTAS for the special case $k=\Omega(n)$ and $m=\Theta(n^{2})$.  
Two interesting variations of the {\it DkS} problem were introduced by Andersen and Chellapilla \cite{AndersenChellapilla}.
The two problems ask for the set $S$ that maximizes the density subject to $s\leq k$  (DamkS) and  $s \geq k$ (DalkS).
They provide a practical 3-approximation algorithm for the DalkS problem
and a slower 2-approximation algorithm. However it is not known whether 
DalkS is \NPhard. For the DamkS problem they showed that if  there exists
a $\gamma$-approximation algorithm for DamkS, then there is a $4(\gamma^2+\gamma)$-approximation
algorithm for the {\it DkS} problem, which indicates that DamkS is likely to be hard as well. 
This hardness conjecture was proved by Khuller and Saha \cite{khuller}.

\spara{Quasi-cliques.}
A set $S \subseteq V$ is a $\alpha$-quasiclique if  $e(S) \geq \alpha {|S| \choose 2}$, i.e., 
if the edge density $f_e(S)$ exceeds a threshold parameter $0 < \alpha \leq 1$. 
Abello et al.  \cite{abello} propose an algorithm for finding maximal 
quasi-cliques. Their algorithm starts with a random vertex and at every step it adds a new vertex
to the current set $S$ as long as the density of the induced graph exceeds the prespecified threshold $\alpha$. 
Vertices that have many neighbors in $S$ and many other neighbors that can also extend $S$ are preferred. 
The algorithm iterates until it finds a maximal $\alpha$-quasi-clique.
Uno presents an algorithm to enumerate all $\alpha$-pseudo-cliques \cite{uno}. 

Recently, \cite{tsourakakis2013denser} introduced a general framework for 
dense subgraph extraction and proposed the optimal quasi-clique problem 
for extracting compact, dense subgraphs. The optimal quasi-clique problem
is \NPhard and inapproximable too \cite{mythesis}.

\spara{$k$-Core.} A $k$-degenerate graph $G$ is a graph in which every subgraph has a vertex of degree at most k.
The degeneracy of a graph is the smallest value of $k$ for which it is $k$-degenerate.
The degeneracy is more known in the graph mining community as the $k$-core number. 
A $k$-core  is a maximal connected subgraph of $G$ in which all vertices have degree at least $k$.
There exists a linear time algorithm for finding $k$ cores by 
repeatedly removing the vertex of the smallest degree \cite{batagelj}.
A closely related concept is the triangle $k$-core, a maximal induced
subgraph of $G$ for which each edge participates in at least $k$ triangles \cite{zhang2012extracting}.
To find a triangle $k$-core, edges that participate in fewer than $k$ triangles
are  repeatedly removed.

\spara{$k$-clubs, $kd$-cliques.} A subgraph $G(S)$ induced by the vertex set $S$ is a $k$-club if the diameter of 
$G(S)$ is at most $k$ \cite{mokken}. $kd$-cliques are conceptually very close to $k$-clubs. 
The difference of a $kd$-clique from a $k$-club is that shortest paths between pairs of vertices from $S$
are allowed to include vertices from $V \backslash S$.

\spara{Shingling.}  Gibson, Kumar and Tomkins \cite{gibson} propose techniques to identify dense bipartite subgraphs
via recursive shingling, a technique introduced by Broder et al. \cite{broder}. 
This technique is geared towards large subgraphs and is based 
on min-wise independent permutations \cite{BroderFrieze}.

\spara{Triangle dense decompositions.} Recently Gupta, Roughgarden 
and Seshadri prove constructively that when the graph has a constant transitivity ratio
then the graph can be decomposed into disjoint dense clusters of radius at most 
two, containing a constant fraction of the triangles of $G$ \cite{gupta2014decompositions}.
 
\subsection{Triangle Counting}
\label{subsec:trianglecounting}

The state of the art algorithm for {\em exact} triangle counting 
is due to Alon, Yuster and Zwick \cite{alon1997finding}  
and runs in $O(m^{\frac{2\omega}{\omega+1}})$,
where currently the fast matrix multiplication exponent $\omega$ is 2.3729
\cite{williams2012multiplying}. Thus, their algorithm currently runs in $O(m^{1.4081})$ time. 
It is worth outlining that algorithms based on matrix multiplication 
are not practical even for medium sized networks due to the high memory requirements.
For this reason, even if listing algorithms solve a more general problem than counting triangles, they
are preferred for large graphs. 
Simple representative algorithms are the node- and the edge-iterator algorithms.
In the former, the algorithm counts for each node the number of edges among its neighbors,
whereas the latter counts for each edge $(i,j)$ the common neighbors 
of nodes' $i,j$. Both have the same asymptotic complexity $O(mn)$, 
which in dense graphs results in $O(n^3)$ time, the complexity of the naive counting algorithm. 
Practical improvements over this family of algorithms have been achieved using 
various techniques, such as hashing and sorting by the degree \cite{latapy2008main,schank2005finding}.
The best known listing algorithm until recently was due to Itai and Rodeh \cite{itai1978finding} 
which runs in $O(m^{3/2})$ time. Recently, Bj\"{o}rklund, Pagh, Williams and Zwick 
gave refined algorithms which are output sensitive algorithms \cite{listing}. 
Finally, it is worth mentioning that a large set of fast approximate triangle counting methods 
exist, e.g., \cite{bar2002counting,braverman2013hard,buriol2006counting,cormode2014second,jha2013space,hu2013massive,DBLP:journals/im/KolountzakisMPT12,pagh2012colorful,pavan2013counting,suri2011counting,tsourakakis2009doulion,tsourakakis2011triangle}.

\section{Problem Definition}
\label{sec:problemdfn} 
In this Section we define and motivate the main problem we consider in this work. 
We first define formally the notion of average triangle density.

\begin{definition}[\mTDS] 
Let $G(V,E)$ be an undirected graph. For any $S\subseteq V$
we define its triangle density $\tau(S)$ as 

$$\tau(S) = \frac{t(S)}{s},$$

\noindent where $t(S)$ is the number of triangles induced by 
$S$ and $s=|S|$.
\end{definition} 

\noindent Notice that $3\tau(S)$ is the average number of (induced) triangles 
per vertex in $S$. In this work we discuss the following problems 
which extend the well-known \DSP \cite{Char00,Goldberg84,Kannan,khuller}.

\begin{problem}[\TDSP] 
Given $G(V,E)$, find a subset of vertices $S^*$ such that $\tau(S^*)=\tau^*_{G}$ 
where 
$$\tau^*_{G}=\max_{S \subseteq V} \tau(S).$$
\end{problem}

\noindent We will omit the index $G$ whenever it is obvious to which graph we refer to.

It is clear that the \DSP and \TDSP in general can result in significantly different
solutions. Consider for instance a graph $G$ on $2n+3$ vertices 
which is the union of a triangle $K_3$  and of a bipartite clique $K_{n,n}$. 
The optimal solutions of the \DSP and the \TDSP are the bipartite clique 
and the triangle respectively. Therefore, the interesting question is whether 
maximizing the average degree and the triangle density result 
in different results in real-world networks. 

\begin{figure*}[htp]
\centering
\includegraphics[width=0.70\textwidth]{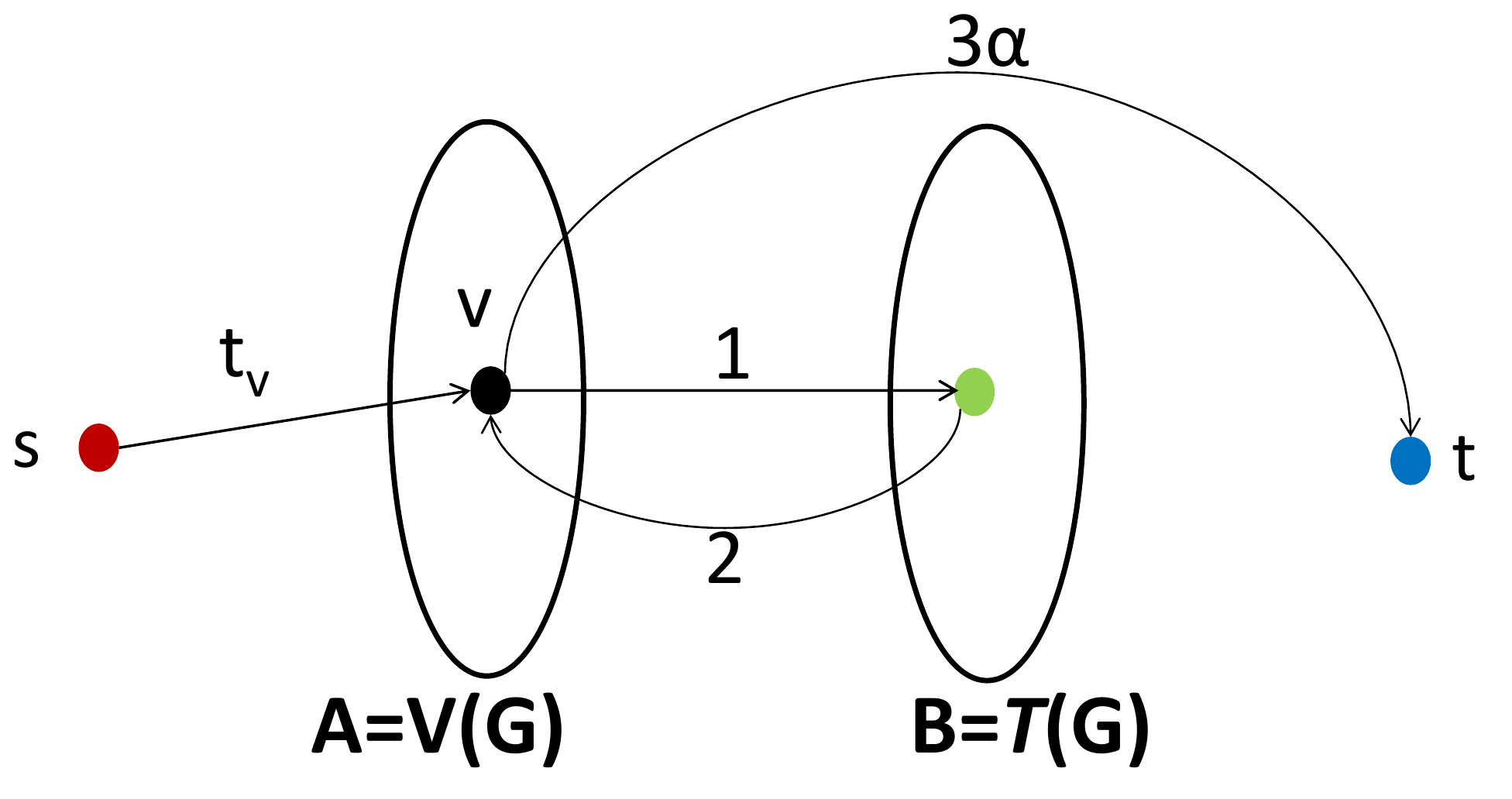}
\caption{\label{fig:network} Figure shows the network $H$ that Algorithm~2
outputs, given a graph $G$ and a parameter $\alpha>0$ as its input. 
Set $A$ corresponds to the vertex set $V(G)$, whereas each vertex in set $B$  corresponds
to a triangle in set $\mathcal{T}(G)$, the set of all triangles in $G$.}
\end{figure*}

It is a well-known fact that triangles play a key role in 
numerous applications related to community detection and 
clustering, e.g., \cite{gleich,watts1998collective}. 
This is reflected to the outcome of the \TDSP. 
Specifically, we observe that solving the \TDSP problem results
in sets with a structure close to a clique.  
This is an important results for applications: as we have mentioned
earlier, there exists a dichotomy among the various formulations used to extract 
dense subgraphs. Either the resulting optimization problem 
is \NPhard or it is polynomially time solvable but tends to output 
subgraphs of larger size than the desired, which fail to be near-cliques.

As we will see in Section~\ref{sec:experiments} in detail, 
the \TDSP consistently succeeds in finding near-cliques, even in cases where 
the \DSP fails. Furthermore, even when the \DSP succeeds in finding dense, compact
subgraphs, the \TDSP output is always superior in terms of the edge 
density $f_e= e(S) / {|S| \choose 2}$
and triangle density $f_t = t(S) / {|S| \choose 3}$\footnote{We will use the term triangle 
density for both $\tau(S)$ and $f_t$. 
It will always be clear from the notation to which of the two measures
we are referring at.}.

Table~\ref{tab:motivational} shows  the results of the optimal subgraphs 
for the \DSP and \TDSP respectively on some popular real-world networks. 
The results are representative on what we have observed on numerous datasets
we have experimented with: 
{\em the \TDSP optimal solution compared to the \DSP optimal solution is 
a smaller and tighter/denser subgraph which exhibits a strong near-clique
structure.} Therefore, the \TDSP appears to combine the best of both worlds: 
polynomial time solvability and extraction of near-cliques.
On the other hand, since all algorithms we propose for the \TDSP require 
running first either a triangle listing or counting algorithm, 
we suggest that the \TDSP should be used in place of the \DSP, 
when the latter fails to extract a near-clique, as the former is computationally
more expensive.

\section{Proposed Method} 
\label{sec:proposed} 
Section~\ref{subsec:exactsols} provides three algorithms which solve \TDSP exactly. 
Sections~\ref{subsec:peel} and~\ref{subsec:mapreduce} provide a $\frac{1}{3}$-approximation algorithm for
the \TDSP and an efficient \mrc implementation respectively. 
Finally, Section~\ref{subsec:kclique} provides a generalization 
of the \DSP and the \TDSP to maximizing the average $k$-clique density 
and shows how the results from previous Sections adapt to this problem.

\subsection{Exact Solutions}
\label{subsec:exactsols}

The algorithm presented in Section~\ref{subsub:maxflow} achieves currently 
the best running time. We present an algorithm which relies on the supermodularity property
of our objective in Section~\ref{subsub:super}.  It is worth outlining that even if 
the algorithm in Section~\ref{subsub:super} is slower, it requires less space than 
the algorithm in Section~\ref{subsub:maxflow}. 
Section~\ref{subsub:lp} presents a linear programming approach which generalizes 
Charikar's linear program \cite{Char00} to the \TDSP.  
Future improvements in the running time of procedures we use as black boxes, 
will imply improvements for our algorithms as well. 

\subsubsection{An $O\big( m^{3/2}+ nt+\min{(n,t)}^3 \big)$-time exact solution}
\label{subsub:maxflow}
\begin{algorithm}[ht]
\begin{algorithmic}[1]
\caption{\tt{\TDS}$(G)$}
\label{alg:alg1}
\STATE{$l \leftarrow 0, u\leftarrow n^3, S^*\leftarrow \emptyset$}
\STATE{List the set of triangles $\mathcal{T}(G)$} 
\WHILE{$u\geq l +\frac{1}{n(n-1)}$} 
\STATE{ $\alpha \leftarrow \frac{l+u}{2}$}
\STATE{ $H_{\alpha} \leftarrow$ \textsf{Construct-Network}$(G,\alpha, \mathcal{T}(G) )$ } 
\STATE{$(S,T)\leftarrow$ min $st$-cut in $H_{\alpha}$}
 \IF{$S=\{s\}$ } 
		\STATE{$u\leftarrow \alpha$} 
 \ELSE
 \STATE{$l\leftarrow \alpha$} 
 \STATE{$S^* \leftarrow  \big( S\backslash \{s\} \big) \cap V(G) $} 
 \ENDIF
\STATE{Return $S^*$}
\ENDWHILE
\end{algorithmic}
\end{algorithm}

 \begin{algorithm}[ht]
\begin{algorithmic}[1]
\caption{ \tt{Construct-Network} $(G,\alpha,   \mathcal{T}(G)  )$ }
\label{alg:alg2}
\STATE{$V(H) \leftarrow \{s\} \cup V(G) \cup \mathcal{T}(G) \cup \{t\} $.}
\STATE{For each vertex $v\in V(G)$ add an arc of capacity 1 to each triangle $t_i$ it participates in.}
\STATE{For each triangle $\Delta=(u,v,w) \in \mathcal{T}(G)$ add arcs to $u,v,w$ 
of capacity 2.} 
\STATE{Add directed arc $(s,v) \in A(H)$ of capacity $t_v$ for each $v\in V(G)$. }
\STATE{Add weighted directed arc $(v,t) \in A(H)$ of capacity $3\alpha$ for each $v\in V(G)$. }
\STATE{Return network $H(V(H),A(H),w), s,t \in V(H)$. }
\end{algorithmic}
\end{algorithm}

\noindent Our main theoretical result is the following theorem.  Its proof is constructive.  

\begin{theorem}
\label{thrm:thrm1} 
There exists a polynomial time algorithm which runs in $O\big( m^{3/2}+nt+\min{(n,t)}^3 \big)$ time,
where $n,t$ are the number of vertices and triangles in graph $G$ respectively,
which solves the  \TDSP in polynomial time. 
\end{theorem} 

\noindent We outline that the first term $O(m^{3/2})$ comes from using 
the Itai-Rodeh \cite{itai1978finding} as our triangle listing blackbox. 
If for instance we use the naive $O(n^3)$ triangle listing algorithm 
then the running time expression is simplified to $O( n^3+nt)$.
On the other hand, if we use the algorithms of Bj\"{o}rklund et al. \cite{listing} 
the first term becomes for dense graphs 
$\tilde{O}\big(n^{\omega}+ n^{3(\omega-1)/(5-\omega)}t^{2(3-\omega)/(5-\omega)}\big)$ 
and for sparse graphs 
$\tilde{O}\big(m^{2\omega/(\omega+1)}+ m^{3(\omega-1)/(\omega+1)}t^{(3-\omega)/(\omega+1)}\big)$,
where $\omega$ is the matrix multiplication exponent. Currently $\omega<2.3729$ 
due to \cite{williams2012multiplying}.
We maintain \cite{itai1978finding}
as our black-box to keep the expressions simpler. However, the 
reader should keep in mind that the result presented in \cite{listing} 
improves the total running time of the first term. 

We will work our way to proving Theorem~\ref{thrm:thrm1} by proving first the following key lemma. 
Then, we will remove the logarithmic factor.

\begin{lemma}
\label{lem:keylem} 
Algorithm 1  solves the  \TDSP in $O\big( m^{3/2}+ ( nt+\min{(n,t)}^3 )\log(n)\big)$ time.
\end{lemma}

Algorithm~1 uses  maximum flow computations to solve \TDSP. It is worth outlining 
that Goldberg's maximum flow method \cite{Goldberg84} for the \DSP does not adapt to the case of \TDSP.
Algorithm~1 returns an optimal subgraph $S^*$, i.e., $\tau(S^*)=\tau^*_{G}$.
The algorithm  performs a binary search on the triangle density value $\alpha$. 
Specifically, each binary search query corresponds to 
querying {\it does there exist a set $S\subseteq V$ such that 
$t(S)/|S| >\alpha$?}. For each binary search, we construct a network $H$ 
by invoking Algorithm~2.  Let $\mathcal{T}(G)$ be the set of triangles in $G$. 
Figure~\ref{fig:network} illustrates this network. The
vertex set of $H$ is $V(H) = \{s\} \cup A \cup B \cup \{t\}$,
where $A=V(G)$ and $B=\mathcal{T}(G)$. For the purpose
of finding $\mathcal{T}(G)$, a triangle listing algorithm is required \cite{listing,itai1978finding}. 
The arc set of graph $H$ is created as follows. 
For each vertex $r \in B$ corresponding to triangle $\Delta(u,v,w)$ we add 
three incoming and three outcoming arcs. 
The incoming arcs come from the vertices $u,v,w \in A$ which form  triangle $\Delta(u,v,w)$. 
Each of these arcs has capacity equal to 1. The outgoing arcs go to the same set 
of vertices $u,v,w$, but the capacities are equal to 2. 
In addition to the arcs of capacity 1 from each vertex $u \in A$ to the triangles it participates in, 
we add an outgoing arc of capacity $3\alpha$ to the sink
vertex $t$. From the source vertex $s$ we add an outgoing arc to each $u \in A$ 
of capacity $t_v$, where $t_v$ is the number of triangles vertex $v$ participates in $G$. 
As we have already noticed,  $H$ can be constructed in $O(m^{3/2})$ time \cite{itai1978finding}. 
It is worth outlining that after computing $H$ for the first time, 
subsequent networks need to update only the arcs that depend on the parameter
$\alpha$, something not shown in the pseudocode for simplicity. 
To prove that Algorithm~1 solves the \TDSP and runs in $O\big( m^{3/2}+( nt+\min{(n,t)}^3 )\log(n)\big)$  
time we will proceed in steps. 

For the sake of the proof, we introduce the following notation. 
For a given set of vertices $S$ let $t_i(S)$ be the number of 
triangles that involve exactly $i$ vertices from $S$, $i \in \{1,2,3\}$.
Notice that $t_3(S)$ is the number of induced triangles by $S$, for which we have been
using the simpler notation $t(S)$ so far. 

We use the following claim as our criterion to set the initial values $l,u$ 
in the binary search. \\
\underline{Claim 1} $0\leq \tau(S) < n^3$ for any $S\subseteq V$. \\
The lower bound is trivial. The upper bound also follows trivially by observing that 
$t_3(S) \leq {n \choose 3}$ and $|S| \geq 1$ for any $\emptyset \neq S \subseteq V$.
This suggests that the optimal value $\tau^*$ is always less than  $n^3$.  


\begin{figure*}[!htp]
\centering
\begin{tabular}{@{}c@{}@{\ }|c@{}@{\ }|c@{}}
\includegraphics[width=0.18\textwidth]{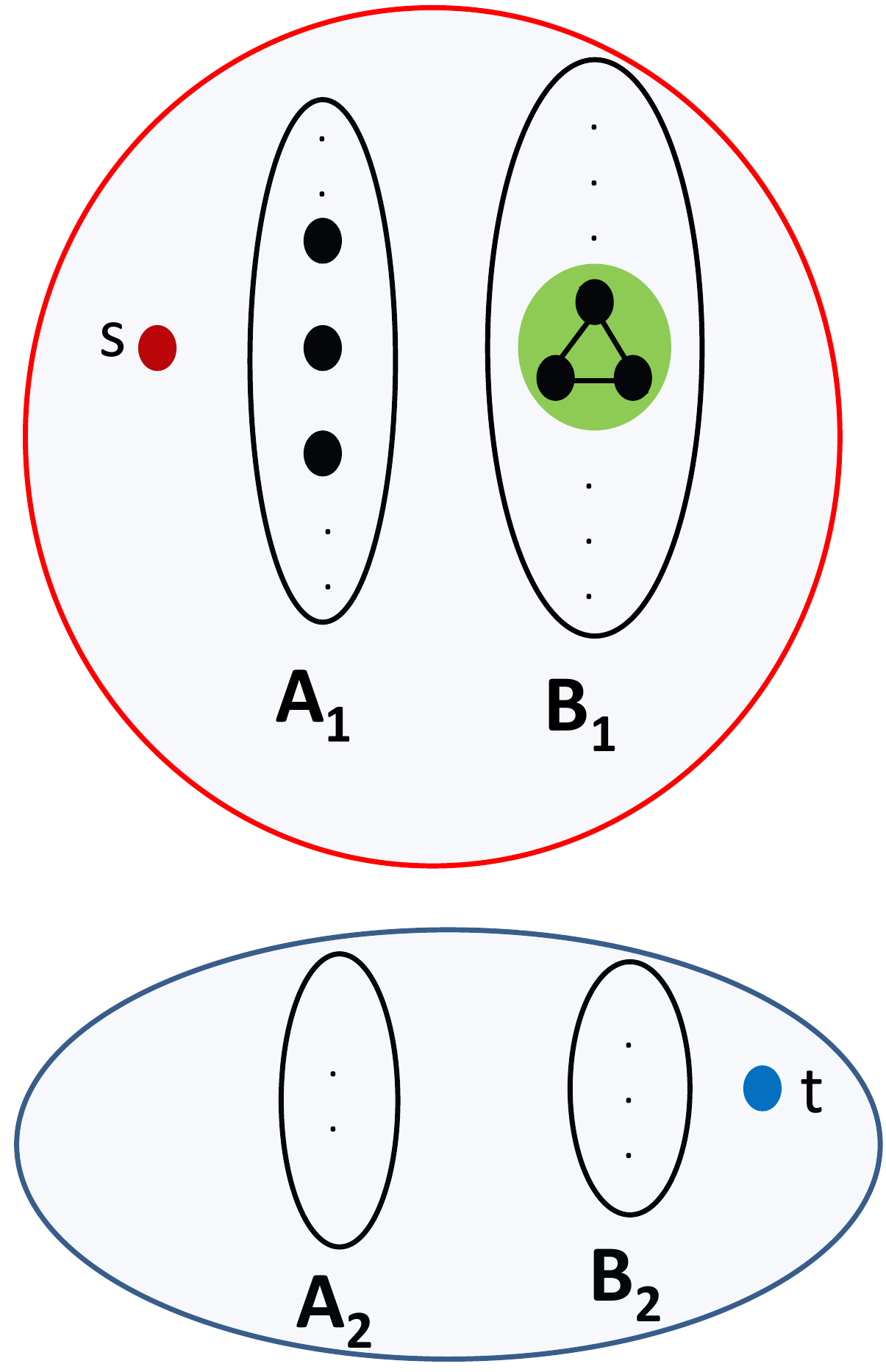} & \includegraphics[width=0.385\textwidth]{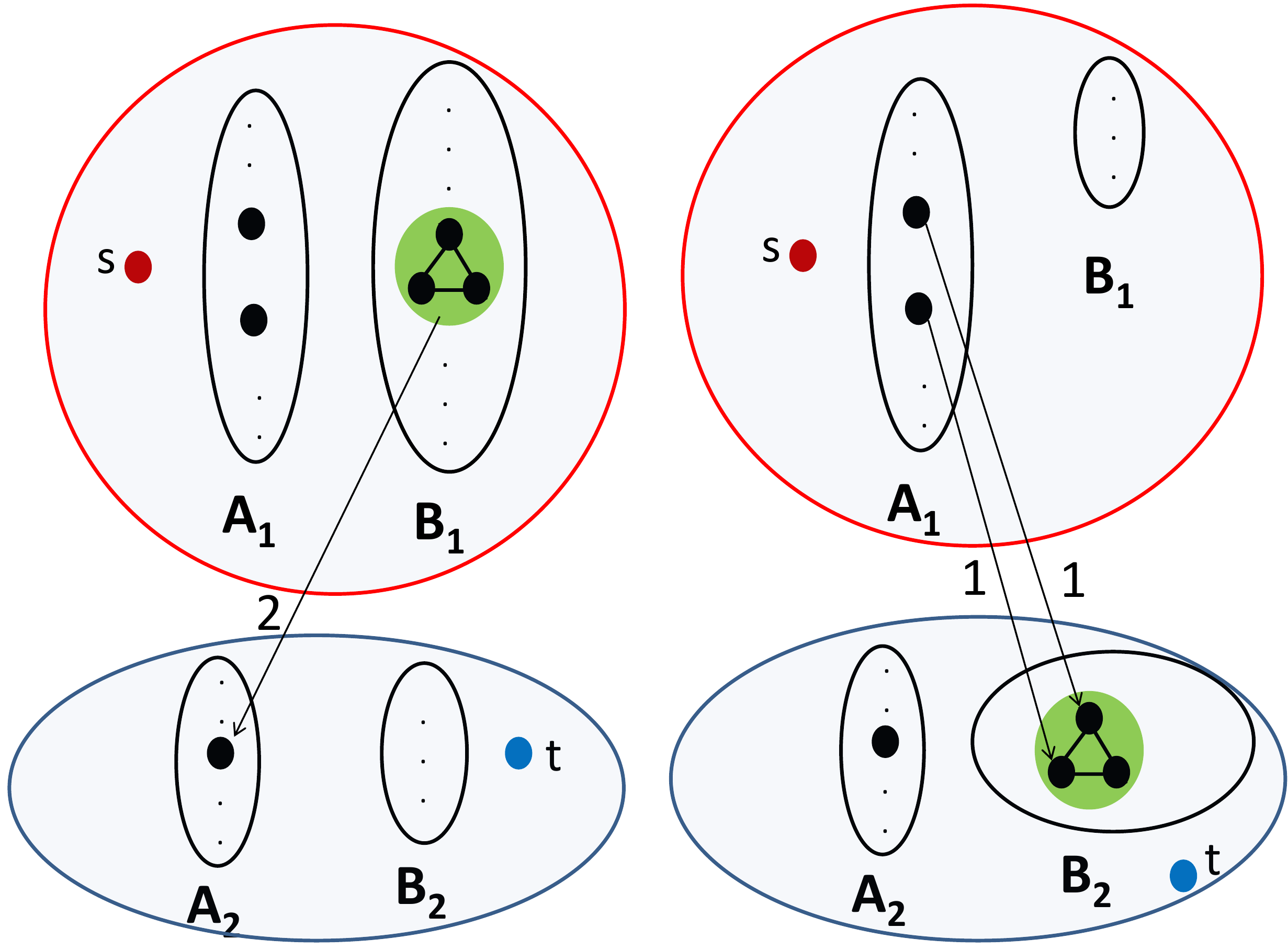}    & \includegraphics[width=0.21\textwidth]{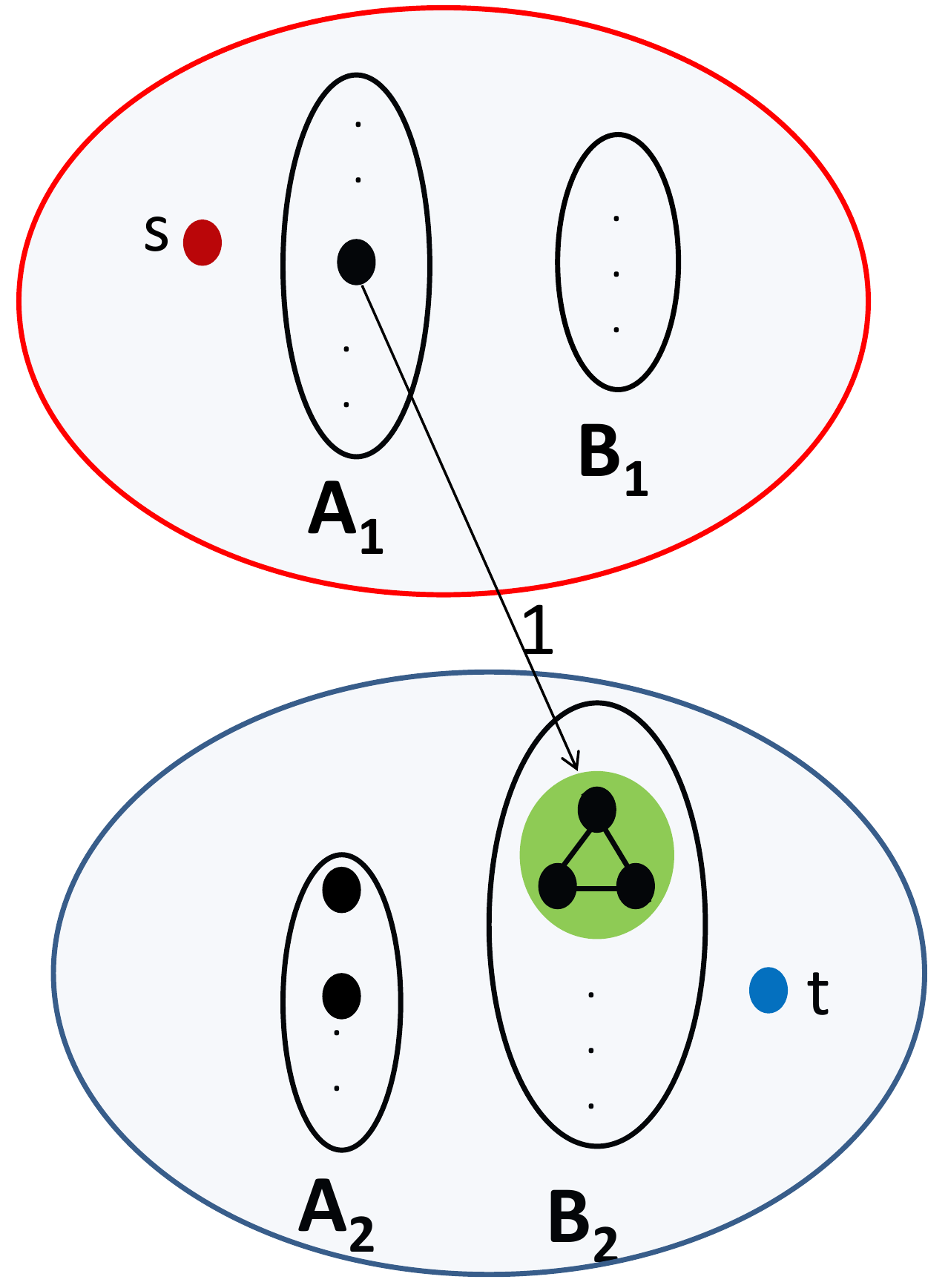}       \\
(a) & (b) & (c)
\end{tabular}
\caption{\label{fig:triangletypes}
Figure shows the structure of any minimum $st$-cut $(S,T)$ in network $H_{\alpha}$, $\alpha>0$.
Sets $S,T$ are shown within the ovals on the top (red) and the bottom (blue) respectively. Figure shows
the structure with respect to three vertices $u,v,w \in A$ (in black) 
which form a triangle and the corresponding triangle-vertex $r \in B$ (green) for the case of a
(a) {\it type 3}, (b) {\it type 2} and (c) {\it type 1} triangle. }
\end{figure*} 
\hide{
Let $u,v,w$ be three vertices which form a triangle $(u,v,w)$. 
In a minimum $st$ cut in network $H_{\alpha}$, for any $\alpha>0$ 
the followings holds. (a){\it Type 3 triangles:} If  $u,v,w \in A\cap S$,
the  vertex $(u,v,w)\in B$ corresponding to the triangle they form, lies in $B\cap S$.
Otherwise the cost of the cut would increase by 3. (b) {\it Type 2 triangles:} 
If $u,v \in A\cap S, w\in A\cap T$ then the triangle-vertex $(u,v,w) \in B$ can lie either
in the side of $S$ or $T$ but in both cases we pay a cost of 2. (c) {\it Type 1 triangles:} 
If $u \in A\cap S, v,w\in A\cap T$ then the triangle-vertex $(u,v,w) \in B$ has to lie
in $B\cap T$, otherwise the cost of the cut would increase by 3.}

The next claim serves as a criterion to decide when to stop the binary search. \\
\underline{Claim 2} The smallest possible difference among two 
distinct values $\tau(S_1),\tau(S_2)$ is equal to $\frac{1}{n(n-1)}$.\\
To see why, notice that the difference $\delta$ between two possible different triangle
density values is 

$$ \delta = \frac{ t(S_1)|S_2|-t(S_2)|S_1| }{|S_1||S_2|}.$$ 

If $|S_1|=|S_2|$ then $|\delta| \geq \frac{1}{n} > \frac{1}{n(n-1)}$,
otherwise $|\delta| \geq \frac{1}{|S_1||S_2|} \geq \frac{1}{n(n-1)}$.
Notice that combining the above two claims shows that the binary 
search terminates in at most $5\log{n}$ queries. The following 
lemma is the key lemma for the correctness of Algorithm~1. 

\begin{lemma} 
\label{lem:stcut}
Consider any $st$ min-cut $(S,T)$ in network $H_{\alpha}$. Let $A_1=S\cap A, B_1=S\cap B$
and $A_2=T\cap A, B_2=T\cap B$. The cost of the min-cut is equal to 
$$\sum_{v \notin A_1}t_v  + 2t_2(A_1)+t_1(A_1) + 3\alpha |A_1|.$$
\end{lemma}

\begin{proof}

\underline{Case I: $A_1 = \emptyset$:}
In this case the proposition trivially holds,
as the cost is equal to $\sum_{v \in A} t_v = 3t$.
It is worth noticing that in this case $B_1$ has to be also empty,
otherwise we contradict the optimality of $(S,T)$. Hence 
$S=\{s\}, T=A \cup B \cup \{t\}$.

\underline{Case II: $A_1 \neq \emptyset$:}

Consider the cost of the arcs from $A_1 \cup B_1$  to 
$A_2 \cup B_2$. We consider three different subcases,
which are illustrated in Figure~\ref{fig:triangletypes}
If there exist three vertices $u,v,w \in A_1$ that form
triangle $\Delta(u,v,w)$, then the vertex $r \in B$ corresponding
to this specific triangle has to be in $B_1$. If not, then $r \in B_2$, 
and we could reduce the cost of the min-cut by 3, if we move the triangle 
to $B_1$. Therefore the cost we pay for triangles of type three is 0. 
This is shown in Figure~\ref{fig:triangletypes}(a). 
Consider three vertices $u,v,w$ such that they form a triangle $\Delta(u,v,w)$ and
$u,v \in A_1, w \in A_2$. Then, the vertex $r \in B$ corresponding to this triangle
can be either in $B_1$ or $B_2$. The crucial point is that 
we always pay 2 in the cut for each triangle of type two as Figure~\ref{fig:triangletypes}(b)
shows. Finally, in the case $u,v,w$  form a triangle, $u \in A_1, v,w \in A_2$
the vertex $r\in B$ corresponding to triangle $\Delta(u,v,w)$ will be in $B_2$.
If not, then it lies in $B_1$ and we could decrease the cost of the cut by 3
if we move it in $B_2$. Hence, we pay 1 in the cut
for each triangle of type one as Figure~\ref{fig:triangletypes}(c) shows.
Therefore the total cost is equal to $2t_2(A_1)+t_1(A_1)$.

Furthermore, the cost of the arcs from source $s$ to $T$ is equal to 
$\sum_{v \in A_2}t_v=\sum_{v \notin A_1}t_v$. The cost of the arcs from $A_1$ to $T$ 
is equal to $3\alpha |A_1|$. 
Summing up the individual cost terms, we obtain that the cost is equal to 
$\sum_{v \notin A_1}t_v  + 2t_2(A_1)+t_1(A_1) + 3\alpha |A_1|$.
\end{proof}

The next lemma proves the correctness of the binary search in Algorithm~1. 

\begin{lemma} 
\label{lem:binary}
(a) If there exists a set $W \subseteq V(G)$ in $G$ such that $t_3(W) > \alpha|W|$ then 
any $st$-min-cut (S,T) in $H_{\alpha}$ satisfies $S\backslash \{s\} \neq \emptyset$. 
(b) Furthermore, if there does not exists a set $W$ such that $t_3(W) > \alpha|W|$ then 
the cut $(\{s\}, A \cup B \cup \{t\})$ is a minimum $st$-cut. 
\end{lemma} 

\begin{proof} 

(a)  Let $W \subseteq V$ be such that

\begin{equation}\label{eq:contradiction}
  t_3(W) > \alpha|W|.
\end{equation}

Suppose for the sake of contradiction that the minimum
$st$-cut is achieved by $(\{s\}, A \cup B \cup \{t\})$. 
In this case the cost of the minimum $st$-cut is 
$\sum_{v \in A}t_v = 3t$. 
Now, consider the following $(S,T)$ cut.
Set $S$ consists of the source vertex $s$, $A_1=W$ and $B_1$ be the 
set of triangles of type 3 and 2 induced by $A_1$. Let $T$ be the 
rest of the vertices in $H$. The cost of this cut is

$$cap(S,T)  = \sum_{v \notin A_1} t_v + 2t_2(A_1)+t_1(A_1) + 3\alpha |A_1|.$$

\noindent Therefore, by our assumption that the minimum
$st$-cut is achieved by $(\{s\}, A \cup B \cup \{t\})$ we obtain 

\begin{equation}
  3t \leq  \sum_{v \notin A_1} t_v + 2t_2(A_1)+t_1(A_1) + 3\alpha |A_1|.
\end{equation}

\noindent Now, notice that by double counting 

$$\sum_{v \in A_1} t_v = 3t_3(A_1)+2t_2(A_1)+t_1(A_1).$$

\noindent  Furthermore, we observe 
$$\sum_{v \in A_1} t_v + \sum_{v \notin A_1} t_v = 3t.$$

\noindent By combining these two facts, and the fact that 
$3t$ is the capacity of the minimum cut,  we obtain  the
following contradiction of Inequality~\eqref{eq:contradiction}. 

\begin{align*}
 3t \leq \sum_{v \notin A_1} t_v + 2t_2(A_1)+t_1(A_1) + 3\alpha |A_1|&\Leftrightarrow  t_3(W) \leq \alpha|W|. 
\end{align*}

(b)
By Lemma~\ref{lem:stcut}, for any minimum $st$-cut $(S,T)$ the capacity of the cut is equal to 
$\sum_{v \notin A_1} t_v + 2t_2(A_1)+t_1(A_1) + 3\alpha |A_1|$ where $A_1 = A \cap S, A_2 = A \cap T$.
Suppose for the same of contradiction that the cut $(\{s\},  A \cup B \cup \{t\})$ 
is not a minimum cut. Therefore, 

$$ cap(\{s\},  A \cup B \cup \{t\}) = 3t> \sum_{v \notin A_1} t_v + 2t_2(A_1)+t_1(A_1) + 3\alpha |A_1|.$$
    
Using the same algebraic analysis as in (a), the above statement implies the contradiction 
$ t_3(W) > \alpha |W|$, where $W=A_1$. 
\end{proof}

\noindent Now we can complete the proof of Lemma~\ref{lem:keylem}.

\begin{proof}
The termination of Algorithm 1 follows directly from Claims 1, 2.
The correctness follows from Lemmata~\ref{lem:stcut},~\ref{lem:binary}.
The running time follows from Claims 1,2 which show that the number 
of binary search queries is $O(\log(n))$ and each binary search query 
can be performed in $O\big( nt+\min{(n,t)}^3 \big)$ time using
the algorithm due to Ahuja, Orlin, Stein and Tarjan \cite{ahuja1994improved}\footnote{Notice that the 
network $H_{\alpha}$ has $O(n+t)$ arcs, therefore the running time 
is $O(\min{(n,t)} (n+t)+\min{(n,t)}^3)=O(nt+\min{(n,t)}^3)$.}
or Gusfield's algorithm \cite{gusfield1991computing}.
\end{proof}

The proof of Theorem~\ref{thrm:thrm1} follows Lemma~\ref{lem:keylem}
and the fact that the parametric maximum flow algorithm of
Ahuja, Orlin, Stein and Tarjan \cite{ahuja1994improved}, see also \cite{GGT89},
saves the logarithmic factor from the running time.


\subsubsection{An $O\big( (n^5 m^{1.4081} + n^6)) \log(n)  \big)$-time exact solution}
\label{subsub:super}
In this Section we provide a second  exact algorithm for the \TDSP. 
First, we provide the necessary theoretical background. 

\begin{definition}[Supermodular function]
Let $V$ be a finite set. The set function $f:2^V\rightarrow \field{R}$ 
is supermodular if and only if for all $A,B \subseteq V$ 

$$ f(A\cup B) \geq f(A)+f(B)-f(A\cap B).$$

\noindent A function $f$ is supermodular if and only if $-f$ is submodular. 
\end{definition}

\noindent Sub- and supermodular functions constitute an important class
of functions with various exciting properties. In this work, we are 
primarily interested in the fact that 
{\em maximizing a supermodular function  is solvable in strongly polynomial time} 
\cite{grotschelgeometric,iwata2001combinatorial,lovasz1983submodular,schrijver2000combinatorial}. 
For our purposes, we state the following result which we use as a subroutine
in our proposed algorithm.

\begin{theorem}[\cite{orlin2009faster}]
\label{thrm:orlins}
There exists an algorithm for maximizing an integer valued supermodular function $f$ 
which runs in $O\big( n^5 EO + n^6)  \big)$ time, where $n=|V|$ is the size
of the ground set $V$ and $EO$ is the maximum amount of time 
to evaluate $f(S)$ for a subset $S\subseteq V$.
\end{theorem}

We show in the following that when the ground set is the set of vertices $V$ 
and $f_{\alpha}:2^V \rightarrow \field{R}$ is defined by $f_{\alpha}(S) = t(S)-\alpha |S|$
where $\alpha \in \field{R}^+$, we can solve the \TDSP in polynomial time.

\begin{theorem}
\label{thrm:supermod}
Function $f:V \rightarrow \field{R}$ where $f(S) = t(S)-\alpha |S|$ is supermodular. 
\end{theorem} 

\begin{proof} 
Let $A,B \subseteq V$. Let $t:2^V \rightarrow \field{R}$ be the function 
which for each set of vertices $S$ returns the number of induced triangles $t(S)$. 
By careful counting

$$t(A \cup B) = t(A)+t(B)-t(A\cap B)+ t_1(A: B \backslash A) + t_2(A: B \backslash A),$$ 

\noindent where $t_1(A: B \backslash A),t_2(A: B \backslash A)$ are the number of triangle
with one, two vertices in $A$ and two, one vertices in $B \backslash A$ respectively. 
Hence,  for any $A,B \subseteq V$ 

$$ t(A \cup B)+ t(A\cap B) \geq t(A)+t(B)$$

\noindent and the function $t$ is supermodular. 
Furthermore, for any $\alpha>0$ the function $-\alpha |S|$  is supermodular. 
Since the sum of two supermodular functions is supermodular, 
the result follows. 
\end{proof}

\noindent Theorem~\ref{thrm:supermod} naturally suggests Algorithm~3. 
The algorithm will run in a logarithmic number of rounds. 
In each round we maximize function $f_{\alpha}$ using Orlin's algorithm \textsf{Orlin-Supermodular-Opt}
which takes as input arguments the graph $G$ and the parameter $\alpha>0$. We assume for simplicity
that within the procedure \textsf{Orlin-Supermodular-Opt} function $f$ is evaluated
using an efficient exact triangle counting algorithm \cite{alon1997finding}. 
The algorithm of Alon, Yuster and Zwick \cite{alon1997finding} runs in $O(m^{2\omega/(\omega+1)})$
time where $\omega<2.3729$ \cite{williams2012multiplying}. This suggests the $EO=O(m^{1.4081})$.
The overall running time of Algorithm~3 is $O\big( (n^5 m^{1.4081} + n^6) \log(n)  \big)$.

\begin{algorithm}[ht]
\begin{algorithmic}[1]
\caption{\tt{\TDS}$(G)$ [Supermodularity]}
\label{alg:supermodular}
\STATE{$l \leftarrow 0, u\leftarrow n^3, S^*\leftarrow V$}
\WHILE{$u\geq l +\frac{1}{n(n-1)}$} 
\STATE{ $\alpha \leftarrow \frac{l+u}{2}$}
\STATE{ $(val,S) \leftarrow$ \textsf{Orlin-Supermodular-Opt}$(G,\alpha)$ } 
 \IF{$val<0$ } 
		\STATE{$u\leftarrow \alpha$} 
 \ELSE
 \STATE{$l   \leftarrow \alpha$} 
 \STATE{$S^* \leftarrow  S$} 
 \ENDIF
\STATE{Return $S^*$}
\ENDWHILE
\end{algorithmic}
\end{algorithm}

\subsubsection{A Linear Programming Approach}
\label{subsub:lp}
In this Section we show how to generalize Charikar's linear program, see $\S$2 in \cite{Char00}, to
provide a linear program (LP) which solves the \TDSP.  The main difference 
compared to Charikar's LP is the fact that we introduce a variable $x_{ijk}$ for each
triangle $(i,j,k) \in \mathcal{T}(G)$.  The LP follows.

\medskip
\begin{equation} \label{LP}
\framebox{
\begin{minipage}[b]{0.65\linewidth}
\medskip\par
\begin{equation*}
\mathbf{max}  \qquad \sum_{(i,j,k) \in \mathcal{T}(G)} x_{ijk} 
\end{equation*}
\vskip -15pt
\begin{eqnarray*}
\mbox{{\bf s.t.} }\,\, x_{ijk} \leq y_i &\forall (i,j,k) \in \mathcal{T}(G) \\
\mbox{ }\,\,  x_{ijk} \leq y_j &\forall (i,j,k) \in \mathcal{T}(G) \\
\mbox{ }\,\,  x_{ijk} \leq y_k &\forall (i,j,k) \in \mathcal{T}(G)\\
\mbox{ }\,\,  \sum_i y_i \leq 1 &  \\
\mbox{ }\,\,  x_{ijk} \geq 0 & \forall (i,j,k) \in \mathcal{T}(G) \\
\mbox{ }\,\,   y_i \geq 0 & \forall i \in V(G) \\
\end{eqnarray*}
\end{minipage}\medskip\par}
\end{equation}

\begin{theorem}
\label{thrm:lpopt} 
Let $OPT_{LP}$ be the value of the optimal solution to the LP~\ref{LP}. 
Then, 

$$ \tau_G^*  = OPT_{LP}.$$ 

\noindent Furthermore, a set $S$ achieving triangle density equal to $\tau_G^*$ can 
be computed from the optimal solution to the LP.
\end{theorem}

\begin{proof}
We break the proof of $ \tau_G^*  = OPT_{LP} $ in two cases.
The second case provides a constructive procedure for finding a set $S^*$
which achieves triangle density equal to $\tau_G^*$.

\underline{Case I: $ \tau_G^*  \leq OPT_{LP}$}\\
We will prove a more general statement:
for any $S \subseteq V$, the value of the $LP$ is at least $\tau(S)$. 
We provide a feasible LP solution which achieves an objective value equal to $\tau(S)$. 
Let $y_i = \frac{1}{|S|} \mathbbm{1}(i \in S)$  for each $i \in V$. For each 
triangle $\Delta(i,j,k)$ induced by $S$ let $x_{ijk}=\frac{1}{|S|}$. 
For every other triangle $\Delta(i,j,k)$ set $x_{ijk}=0$.
This is a feasible solution to the LP which achieves an objective value equal to $\frac{t(S)}{|S|}$.
By setting $S=S^*$, we obtain $ \tau_G^*  \leq OPT_{LP}$.

\underline{Case II: $ \tau_G^*  \geq OPT_{LP}$}\\
Let $(\bar{x},\bar{y})$ be the optimal solution to the LP.  
We define $S(r)=\{ i: \bar{y}_i \geq r\}$, $T(r) = \{ \Delta(i,j,k) \in \mathcal{T}(G):  \bar{x}_{ijk} \geq r\}$. 
Notice that since $\bar{x}_{ijk} \leq \min{(\bar{y}_i,\bar{y}_j,\bar{y}_k)}$, 
the inequality $\bar{x}_{ijk} \geq r$ implies that vertices $i,j,k$ belong in set $S(r)$. 
Furthermore, $\int_0^1|S(r)|dr = \sum_{i=1}^n \bar{y}_i \leq 1$ and 
 $\int_0^1|T(r)|dr = \sum_{\Delta{(i,j,k)}}x_{ijk}$. If we assume that there exists no value 
 $r$  such that $|T(r)|/|S(r)| \geq OPT_{LP}$ we obtain the contradiction 
 
 $$ OPT_{LP} = \int_0^1|T(r)|dr < OPT_{LP} \int_0^1|S(r)|dr \leq OPT_{LP}.$$ 
 
 \noindent Hence,  $\tau_G^* \geq OPT_{LP}$. To find a set $S^*$ that achieves 
 triangle density at least $OPT_{LP}$, we need to check at most $n$ different
 values of $r$ and checking the corresponding sets $S(r)$. 
\end{proof} 

\hide{
We break the proof of Theorem~\ref{thrm:lpopt} in two lemmata.

\begin{lemma}
\label{lem:c1}
$$ \tau_G^*  \leq OPT_{LP}.$$ 
\end{lemma}

\begin{proof}
We will prove a more general statement:
for any $S \subseteq V$, the value of the $LP$ is at least $\tau(S)$. 
We provide a feasible LP solution which achieves an objective value equal to $\tau(S)$. 
Let $y_i = \frac{1}{|S|} \mathbbm{1}(i \in S)$  for each $i \in V$. For each 
triangle $\Delta(i,j,k)$ induced by $S$ let $x_{ijk}=\frac{1}{|S|}$. 
For every other triangle $\Delta(i,j,k)$ set $x_{ijk}=0$.
This is a feasible solution to the LP which achieves an objective value equal to $\frac{t(S)}{|S|}$.
By setting $S=S^*$, we obtain $ \tau_G^*  \leq OPT_{LP}$.
\end{proof}
}

\subsection{A $\frac{1}{3}$-approximation algorithm}
\label{subsec:peel} 
In this Section we provide an algorithm for the \TDSP which provides 
a $\frac{1}{3}$-approximation. Our algorithm follows the peeling paradigm,
see \cite{AITT00,Char00,khuller,jiang2013parallel}. Specifically, 
in each round it removes the vertex which participates in the smallest number
of triangles and returns the subgraph that achieves the largest triangle density.
The pseudocode is shown in Algorithm 4.

 \begin{algorithm}[ht]
\begin{algorithmic}[1]
\caption{ \tt{Peel-Triangles}$(G)$ }
\label{alg:alg3}
\STATE{$n\leftarrow |V|, H_n \leftarrow G$}
\FOR{$i\leftarrow n$ to $2$} 
\STATE{Let $v$ be the vertex of $G_i$ of minimum number of triangles} 
\STATE{$G_{i-1} \leftarrow G_i \backslash{v}$}
\ENDFOR
\STATE{Return  $H_j$ that achieves maximum triangle density among $H_i$s, $i=1,\ldots,n$. }
\end{algorithmic}
\end{algorithm}

\begin{theorem}
\label{thrm:3approx}
Algorithm~4 is a $\frac{1}{3}$-approximation algorithm for the \TDSP. 
\end{theorem}

\begin{proof}
Let $S^*$  be an optimal set. Let $v \in S^*, |S^*|=s^*$ and 
$t_A(v)$ be the number of induced triangles by $A$ that $v$ participates in. 
Then, 

\begin{align*}
\tau^*_{G}=\frac{t(S^*)}{s^*} \geq \frac{ t(S^*\backslash \{v\}) }{s^*-1} &\Leftrightarrow t_{S^*}(v) \geq \tau^*_{G}, 
\end{align*}

\noindent since $t(S^*\backslash \{v\}) = t(S^*) - t_{S^*}(v)$.
Consider the iteration before the algorithm removes the first vertex $v$ that belongs in $S^*$. 
Call the set of vertices $W$. Clearly, $S^*\subseteq W$ and for each vertex $u\in W$ the following
lower bound holds $t_W(u) \geq t_W(v)  \geq t_{S^*}(v)> \tau^*_{G}$ due to the greediness of Algorithm~3. 
This provides a lower bound on the total number of triangles induced by $W$

\begin{align*} 
t(W) &= \frac{1}{3} \sum_{u \in W} t_{W}(u) \geq \frac{1}{3}|W| \tau^*_{G} \Rightarrow \frac{t(W)}{|W|} \geq \frac{1}{3} \tau^*_{G}.
\end{align*} 

To complete the proof, notice that the algorithm returns a subgraph $S$ such that  $\tau(S) \geq \tau(W) \geq \frac{1}{3}\tau^*_{G}$.
\end{proof} 

In Section~\ref{subsec:setup}  we provide a simple implementation which runs 
in $O\big(\sum_v {deg(v) \choose 2}\big)=O(mn)$ time with the use of extra space. 
The key differences compared to the \DSP peeling algorithm \cite{Char00},  
are (i) we need to count triangles initially and (b) when 
we remove a vertex, the counts of its neighbors can decrease more than 1
in general. Therefore, when vertex $v$ is removed, we update the counts 
of its neighbors in $O\Big({deg(v) \choose 2}\Big)$ time. 

\hide{
When we remove a vertex from the graph, we

\noindent {\em Implementation Details:} 

A naive implementation of Algorithm 3 which runs in $O(mn)$ time 
and uses $O(n^2)$ space works as follows. 
Initially we count the number
of triangles in $G$, which in general takes $O(mn)$ time \cite{schank2005finding}. 
We maintain an array of size $t_{\max}+1$ where $t_{\max}=\max_{v \in V(G)} t_v$. 
Each entry of the array 

follows. 

We maintain an array 
of hash tables, where the $i$-th entry of the array contains the hash table with all the vertices
with $i$ participating triangles. Notice that at most $O(n)$ such entries suffice as there can be 
at most $n$ different $t_v$ values ranging from 0 up to ${n-1 \choose 2}$. 
 We also maintain the smallest index $i_{\min}$ in the array with a non-empty hash table. 
 When we peel of any vertex from 
}

\subsection{ \textsc{MapReduce} Implementation } 
\label{subsec:mapreduce} 
The \mrc framework \cite{dean} has become the {\it de facto} standard for processing
large-scale datasets. Since the original work of Dean and Ghemawat \cite{dean}, 
a lot of research has focused on developing efficient algorithms for various
graph theoretic problems including the densest subgraph problem \cite{bahmani2012densest}, 
minimum spanning trees \cite{karloff2010model,lattanzi2011filtering}, 
finding connected components \cite{DBLP:conf/icdm/KangTF09,karloff2010model,lattanzi2011filtering} 
and estimating the diameter \cite{Kang:2011:HMR:1921632.1921634},
triangle counting \cite{pagh2012colorful,suri2011counting,tsourakakis2011triangle}
and matchings, covers and  min-cuts \cite{lattanzi2011filtering}.

In the following, we show how we can approximate efficiently the \TDSP in \mrc.
Before we describe the algorithm, we show that Algorithm 5 for any $\epsilon>0$
terminates and provides a $\frac{1}{3+3\epsilon}$-approximation. The idea behind
this algorithm is to peel vertices in batches \cite{bahmani2012densest,goodrich2011external} 
rather than one by one.

 \begin{algorithm}[ht]
\begin{algorithmic}[1]
\caption{ \tt{Peel-Triangles-in-Batches}$(G,\epsilon>0)$ }
\label{alg:alg4}
\STATE{ $S_{out}, S \leftarrow V$} 
\WHILE{ $S \neq \emptyset$}
\STATE{ $A(S) \leftarrow \{ i \in S: t_S(i) \leq 3(1+\epsilon) \tau(S)\} $}
\STATE{ $S \leftarrow S \backslash A(S)$}
\IF{$ \tau(S) \geq \tau(S_{out}) $} 
\STATE{ $S_{out} \leftarrow S$}	
\ENDIF
\ENDWHILE
\STATE{Return  $S_{out}$.}
\end{algorithmic}
\end{algorithm}

\begin{lemma} 
\label{lem:peeloff}
For any $\epsilon>0$, Algorithm~5 provides a $\frac{1}{(3+3\epsilon)}$-approximation to the \TDSP.
Furthermore, it terminates in $O(\log_{1+\epsilon}(n))$ passes. 
\end{lemma} 

\begin{proof*} 
Let $S^*$ be an optimal solution to the \TDSP. 
As we proved in Theorem~\ref{thrm:3approx}, 
for any $v\in S^*$ it is true that $t_{S^*}(v) \geq \tau^*_{G}$.
Furthermore, in each round at least one vertex is removed. To see why, assume
for the sake of contradiction that $A(S)=\emptyset$ for some $S$ during the execution
of the algorithm. Then, we obtain the contradiction that
$3|S|\tau(S)= \sum_{v \in S} t_S(v) \geq (3+3\epsilon)|S|\tau(S)$.
Consider the round where the algorithm for the first time removes a vertex $v \in S^*$.
Let $W$ be the  corresponding set of vertices. 
Since $v \in A(W)$ is peeled off, we obtain an upper bound on its
induced degree, namely  $v \in A(W) \Rightarrow t_W(v) \leq (3+3\epsilon)\tau(W)$. 
Since $S^*\subseteq W$,  we obtain 

\begin{align*} 
(3+3\epsilon) \tau(W) &\geq t_W(v) \geq t_{S^*}(v) \geq \tau(S^*), 
\end{align*}  
\noindent which proves that Algorithm~5 is a $\frac{1}{(3+3\epsilon)}$-approximation to the \TDSP.
To see why the algorithm terminates in logarithmic number of rounds, notice that 

\begin{align*}
3t(S) &> \sum_{v \in S\backslash A(S) } t_{S}(v)  \geq (3+3\epsilon) \big(|S|-|A(S)|\big) \frac{t(S)}{|S|} \Leftrightarrow \\
|A(S)| &\geq \frac{\epsilon}{1+\epsilon} |S| \Leftrightarrow |S\backslash A(S)| \leq \frac{1}{1+\epsilon}|S|.
\end{align*}

\noindent Since $S$ decreases by a factor of $\frac{1}{1+\epsilon}$ in each round, the algorithm
terminates in $O(\log_{1+\epsilon}(n))=O\big(\frac{\log(n)}{\epsilon}\big)$ rounds. 
\end{proof*}

\noindent {\em \mrc Implementation:} Now we are able to describe our algorithm in \mrc. It uses 
any of the efficient algorithms of Suri and Vassilvitski \cite{suri2011counting} 
as a subroutine to count triangles per vertex in each round. 
The removal of the vertices which participate in less triangles than the threshold, 
is done in two rounds, as in \cite{bahmani2012densest}. 
For completeness, we describe the procedure here. 
The set of vertices $S$ to be peeled off in each round are marked by adding a key-value pair $\langle v;\$ \rangle$
for each $v\in S$. Each edge $(u,v)$ is mapped to $\langle u;v \rangle$. 
The reducer receives all endpoints of the edges incident with $v$ and the symbol $\$$ in case the 
vertex is marked for deletion. In case the vertex is marked, then the reduce task returns nothing,
otherwise it copies its input. In the second round, we perform the same procedure with the only difference
being that we map each edge $(u,v)$ to $\langle v;u \rangle$. Therefore, the edges which remain have both endpoints
unmarked. The algorithm runs in  $O(\log(n)/\epsilon)$, as 
it takes  $O(\log(n)/\epsilon)$ peeling off rounds,
and in each peeling round, constant number of rounds is needed to count
triangles per vertex, mark vertices for deletion and remove the corresponding vertex set.

\subsection{$k$-clique Densest Subgraph} 
\label{subsec:kclique}
We outline that our proposed methods can be adapted to the following 
generalization of \DSP and \TDSP.

\begin{definition}[\KCDS] 
Let $G(V,E)$ be an undirected graph. For any $S\subseteq V$
we define its $k$-clique density $h_k(S)$, $k\geq 2$ as 

$$h_k(S) = \frac{c_k(S)}{s},$$

\noindent where $c_k(S)$ is the number of $k$-cliques induced by 
$S$ and $s=|S|$.
\end{definition}

\begin{problem}[\KCDSP] 
Given $G(V,E)$, find a subset of vertices $S^*$ such that $h_k(S^*)=h^*_{k}$ 
where 
$$h^*_{k}=\max_{S \subseteq V} h_k(S).$$
\end{problem}

As in the triangle densest subgraph problem, we create a network $H$ parameterized
by the value $\alpha$ on which we perform our binary search.  
The procedure is described in Algorithm~6. The set $\mathcal{C}(G)$ is the 
set of $k$-cliques in $G$. We then invoke Algorithm~1, with the upper bound $u$ set to $n^k$. 
Following the analysis of Theorem~\ref{thrm:thrm1}, we see that the \KCDSP 
is solvable in polynomial time. For instance, using Gusfield's algorithm 
\cite{gusfield1991computing} or  \cite{ahuja1994improved} in each binary search query we get an overall
running time $O\big(  n^k+ (n|\mathcal{C}(G)| + n^3) \log(n)\big) =O(n^{k+1}\log(n))$.
Using the improved result due to Ahuja, Orlin, Stein and Tarjan for parametric 
max flows in unbalanced bipartite graphs \cite{ahuja1994improved}, 
we save the logarithmic factor in the running time.

 \begin{algorithm}[ht]
\begin{algorithmic}[1]
\caption{ \tt{Construct-Network-$k$} $(G,\alpha,   \mathcal{C}(G),k  )$ }
\label{alg:alg5}
\STATE{$V(H) \leftarrow \{s\} \cup V(G) \cup \mathcal{C}(G) \cup \{t\} $.}
\STATE{For each vertex $v\in V(G)$ add an arc of capacity 1 to each $k$-clique $c_i$ it participates in.}
\STATE{For each $k$-clique $(u_{i_1},\ldots,u_{i_k}) \in \mathcal{C}(G)$ add arcs to $u_{i_1},\ldots,u_{i_k}$ 
of capacity $k-1$.} 
\STATE{Add directed arc $(s,v) \in A(H)$ of capacity $c_v$ for each $v\in V(G)$. }
\STATE{Add weighted directed arc $(v,t) \in A(H)$ of capacity $k\alpha$ for each $v\in V(G)$. }
\STATE{Return network $H(V(H),A(H),w), s,t \in V(H)$. }
\end{algorithmic}
\end{algorithm}

\noindent  Furthermore, Algorithm 4 can also be modified,
by removing in each round the vertex with the smallest 
number of $k$-cliques, to obtain  Corollary~\ref{cor:kapprox}. 
As the analogy of Theorem~\ref{thrm:3approx}.

\hide{ The adaptation is shown in Algorithm~6. 

 \begin{algorithm}[ht]
\begin{algorithmic}[1]
\caption{ \tt{Peel-$k$-cliques}$(G)$ }
\label{alg:alg6}
\STATE{$n\leftarrow |V|, H_n \leftarrow G$}
\FOR{$i\leftarrow n$ to $2$} 
\STATE{Let $v$ be the vertex of $G_i$ of minimum number of $k$-cliques} 
\STATE{$G_{i-1} \leftarrow G_i \backslash{v}$}
\ENDFOR
\STATE{Return  $H_j$ that achieves maximum $k$-clique density among $H_i$s, $i=1,\ldots,n$. }
\end{algorithmic}
\end{algorithm}
}

\begin{corollary}
\label{cor:kapprox}
The algorithm which peels off in each round the vertex with the minimum
number of $k$-cliques and returns the subgraph that achieves the largest 
$k$-clique density, is a $\frac{1}{k}$-approximation algorithm for the \KCDSP. 
\end{corollary}

Similarly, Algoritm~5 and the \mrc implementation can be modified 
to solve the \KCDSP.  We omit the details.

\begin{corollary}
\label{cor:kapprox}
The algorithm which peels off in each round the set of vertices 
with less than $k(1+\epsilon)h(S)$, where $h(S)$ is the $k$-clique
density in that round, terminates in $O(\log_{1+\epsilon}(n))$ rounds
and provides a $\frac{1}{k(1+\epsilon)}$-approximation guarantee for the \KCDSP.
Furthermore, using \cite{finocchi2014counting},  we obtain an efficient \mrc implementation. 
\end{corollary}

\noindent We notice that in general there exist benefits from moving to higher order $k$ values.
Consider the following example which can be further formalized (details omitted). 
Let $G \sim G(n,p)$ be an Erd\"{o}s-R\'enyi graph, where $p=p(n)$. 
Assume that we plant a clique $K$ of size $n^{\gamma}$ for some constant $\gamma>0$.
We wish to show a non-trivial range of $p=p(n)$ values such that the following conditions hold: 

$$ h_2(C) =\frac{|E(K)|}{|K|}=\frac{{n^{\gamma} \choose 2}}{n^{\gamma}} < \frac{p {n \choose 2}}{n}=\Mean{h_2(V)}$$. 

\noindent and for $k\geq 3$

$$ h_k(C) = \frac{{n^{\gamma} \choose k}}{n^{\gamma}} >\frac{p^{{k \choose 2}} {n \choose k}}{n}=\Mean{h_k(V)}$$

\noindent By simple algebraic manipulation we see that $p$ satisfies both conditions if

$$ O\big(n^{-(1-\gamma)}\big)< p < O\big( n^{-\tfrac{2}{k} (1-\gamma)} \big)\footnote{
Notice that for this range of $p$, the graph is connected 
and the clique number is constant with high probability \cite{bollobas}}.$$ 

Clearly, for larger $k$ values, we allow ourselves a larger range of $p$ values for which 
we can find the hidden clique in expectation.
We have implemented the algorithms for \KCDSP but we defer the experimental analysis 
on real graphs for an extended version of this work. Our main finding from preliminary
results with $k=4$, is that in few cases there exists a benefit to maximizing 
the {\em average $K_4$ density}. However, the gain obtained 
by moving from the \DSP to the \TDSP with respect to extracting  
a near-clique is larger than the gain by moving the \TDSP to the \fourCDS.

\begin{table*}[!ht]
\begin{center}
\begin{tabular}{|l|r|r|l|} \hline
   Name                     & Nodes       & Edges     &  Description  \\ \hline 
  \textsf{Adjnoun}          & 112         & 425       &  Generated by processing text data \\ \hline 
  \textsf{AS-735}           & 6\,475      & 12\,572   &  Autonomous Systems \\ \hline
  \textsf{AS-caida}         & 26\,475     & 53\,381   &  Autonomous Systems \\ \hline
  \textsf{ca-Astro}         & 17\,903     & 196\,972  &  Person to Person \\ \hline
  \textsf{ca-GrQC}          & 4\,158      & 13\,422   &  Person to Person \\ \hline
  \textsf{Celegans}         & 297         & 4\,296    &  Neural network of C. Elegans \\ \hline
  \textsf{DBLP}            & 53\,442         & 255\,936   &  Person to Person \\ \hline
  \textsf{Epinions}         & 75\,877     & 405\,739  &  Person to Person \\ \hline
  \textsf{Enron}            & 33\,696     & 180\,811  &  Email \\ \hline 
  \textsf{EuAll}            & 224\,832    & 339\,925  &  Email \\ \hline 
  \textsf{Football}         & 115         & 613       &  NCAA football game network \\ \hline
  \textsf{Karate}           & 34          & 78        &  Person to Person \\ \hline
  \textsf{Lesmis}           & 77          & 254       &  Generated by processing text data \\ \hline
  \textsf{Political blogs}  & 1\,490      & 16\,715   &  Generated by processing sales data \\ \hline  
  \textsf{Political books}  & 105         & 441       &  Blog network \\ \hline 
  \textsf{soc-Slashdot0811} & 77\,360     & 469\,180  &  Person to Person \\ \hline
  \textsf{soc-Slashdot0902} & 82\,168     & 504\,230  &  Person to Person \\ \hline
  \textsf{wb-cs-Stanford}   & 8\,929      & 26\,320   &  Web Graph (page to page) \\ \hline
\end{tabular}
\end{center}
\caption{Datasets used in our experiments.}
\label{tab:datasets}
\end{table*}

\section{Experimental Evaluation}
\label{sec:experiments}

\begin{table*}[!ht]
\begin{center}
\small
\begin{tabular}{|c|c|ccccccc|} \hline
                    Method                 & Measure                  & \textsf{Adjnoun}  & \textsf{Celegans} &  \textsf{Football} & \textsf{Karate} & \textsf{Lesmis} &  \textsf{Polblogs} & \textsf{Polbooks} \\ \hline
                                           & $\frac{|S|}{|V|}$(\%)     & 42.86             & 45.8              &  100               & 47.1            & 29.9            &   19.1                    &  51.4                    \\ 
                                           & $\delta$                  & 9.58              & 17.16             &  10.66             & 5.25            & 10.78           &   55.82                   &  9.40                    \\        
DS                                         & $f_e$                     & 0.20              & 0.13              &  0.094             & 0.35            & 0.49            &   0.196                   &  0.18                    \\
										   & $\tau$                    & 14                & 45.93             &  21.12             & 5.64            & 41.61           &   768.87                  &  22.68                    \\                            
										   & $f_t$                     & 0.013             & 0.005             &  0.003             & 0.05            &  0.18           &   0.019                   &  0.016                    \\ \hline  
                                           & $\frac{|S|}{|V|}$(\%)     & 41.1              & 42.4              &  100               & 52.9            & 29.9            &   18.7                    &  57.1                    \\ 
                                           & $\delta$                  & 9.57              & 17.1              &  10.66             & 5.2             & 10.78           &   55.8                    &  9.3                    \\        
$\frac{1}{2}$-DS                           & $f_e$                     & 0.21              & 0.14              &  0.094             & 0.31            & 0.49            &   0.20                    &  0.16                   \\
										   & $\tau$                    & 14.16             & 46.5              &  21.12             & 5.16            & 41.61           &   774.6                   &  22.68                    \\                            
										   & $f_t$                     & 0.014             & 0.006             &  0.003             & 0.04            &  0.18           &   0.02                    &  0.013                    \\ \hline  										   
                                           & $\frac{|S|}{|V|}$(\%)     & 36.6              & 10.4              &  15.7              & 17.7            & 16.9            &   8.1                    &  19.1                    \\ 
                                           & $\delta$                  & 9.37              & 13.81             &  8.22              & 4.67            & 10.62           &   55.72                   &  9.34                    \\        
TDS                                        & $f_e$                     & 0.23              & 0.46              &  0.48              & 0.93            & 0.89            &   0.46                   &  0.50                    \\
										   & $\tau$                    & 15                & 56.82             &  28                & 8.01            & 47.31           &   972.36                  &  25.95                    \\                            
										   & $f_t$                     & 0.019             &  0.13             & 0.21               &0.80             &  0.72           &   0.136                   &  0.15                    \\ \hline  				
										   & $\frac{|S|}{|V|}$(\%)     & 36.6              & 9.1               &  15.7              & 17.7            & 16.9            &   8.1                    &  15.2                    \\ 
                                           & $\delta$                  & 9.37              & 13.56             &  8.22              & 4.67            & 10.62           &   55.72                   &  9.13                    \\        
$\frac{1}{3}$-TDS                          & $f_e$                     & 0.23              & 0.52              &  0.48              & 0.93            & 0.89            &   0.46                   &  0.61                    \\
										   & $\tau$                    & 15                & 56.55             &  28                & 8.01            & 47.31           &   972.36                  &  25.5                    \\                            
										   & $f_t$                     & 0.019             &  0.17             & 0.21               &0.80             &  0.72           &   0.136                   &  0.24                    \\ \hline  						
\end{tabular}
\caption{\label{tab:motivational} Comparison of the extracted subgraphs by the 
Goldberg's exact algorithm for the \DSP (DS), Charikar's $\frac{1}{2}$-approximation 
algorithm ($\frac{1}{2}$-DS), our exact algorithm for the \TDSP (TDS) and 
our $\frac{1}{3}$-approximation algorithm ($\frac{1}{2}$-TDS). Here, $f_e(S)= e(S) / {|S| \choose 2}$ is the edge density of the extracted subgraph,  
$\delta(S)=2e(S)/|S|$ is the average degree,
$f_t(S) = t(S) / {|S| \choose 3}$ is the triangle density
and $\tau(S)=3t(S)/ |S|$ is the average number of triangles per vertex.
}
\end{center}
\end{table*}

The main goal of this Section is to show that 
the proposed algorithms for the \TDSP constitute new graph mining primitives that can be used 
to find near-cliques when the \DSP fails. 
Additionally to this goal, we compare the quality of the $\frac{1}{3}$-approximation 
algorithm (Algorithm~4) to the optimal algorithm. 
Finally, we explore the trade-off between the approximation quality 
and the number of rounds by ranging the parameter $\epsilon$ in Algorithm~5.

\subsection{Experimental Setup} 
\label{subsec:setup}

The datasets we use are shown in Table~\ref{tab:datasets}.
The experiments were performed on a single machine, with Intel(R) Core(TM) i5 CPU
at 2.40 GHz, with 3.86GB of main memory.  
We have implemented Algorithm~1 in \textsc{Matlab} R2011a using a maximum flow
implementation due to Kolmogorov and Boykov \cite{boykov2004experimental} as our subroutine
which runs in time $O( t(n+t)^3 )$. This implementation can be prohibitively expensive even for 
small graphs which have a large number of triangles.
In the next section we evaluate the exact algorithm on a subset of graphs.

The space usage due to the construction
of the network $H_{\alpha}$ -which has $O(n+t)$ vertices and $O(n+t)$ arcs-
can be large as many networks have a large number of triangles. 
It is worth outlining that when the space usage is a problem whereas the running
time is not, the supermodularity algorithm can be used instead.  
Furthermore, it is also worth mentioning that using any standard maximum flow algorithm rather than the 
algorithm of Ahuja et al. \cite{ahuja1994improved} results in an 
expensive algorithm which runs in $\tilde{\Omega}((n+t)^2)$ time\footnote{
The state-of-the-art algorithm for exact maximum flow is due to Orlin and 
runs in $O(nm)$ where $n,m$ are the number of vertices and edges in the graph. 
In our network we have $O(n+t)$ vertices  and $O(n+t)$ edges with integer 
capacities, resulting in a total $O((n+t)^2)$ time.}.We have written an efficient implementation of 
our peeling algorithm in \textsc{Java JDK 1.6} which runs in $O(nm)$ time.
Our implementation maintains an array of size $O(n)$ containing the counts
of triangles per vertex and an array of at most $O(\max_v t_v)$ entries 
each one pointing to a hash table (notice there exist at most $n$ entries with non-empty hash tables). 
The hash table at position $i$ of the array keeps the set of vertices with exactly $i$ participating 
triangles. At any iteration, we maintain the minimum index
of the array pointing to a non-empty hash table. When we remove a vertex, we update
the triangle counts of its neighbors, move them and place them in the appropriate
hash table if needed, and if one of the updated counts is less than the number of
triangles that the index points at, then we update the index accordingly. 
The total running time is $O\Big({deg(v) \choose 2}\Big)=O(nm)$.
We measure the quality of each extracted subgraph by two measures:   
the edge density of the extracted subgraph $f_e= e(S) / {|S| \choose 2}$
and the triangle density $f_t = t(S) / {|S| \choose 3}$.  
Notice that when $f_e,f_t$ are close to 1, the extracted subgraph is close 
to being a clique.

\subsection{Experiments}
\label{subsec:experiments} 

Table~\ref{tab:motivational} shows the results obtained on several popular
small- and medium-sized graphs. Each column corresponds to a dataset. 
The rows correspond to measurements for each method we use to extract
a subgraph. Specifically, the first (DS), second ($\tfrac{1}{2}$-DS), third (TDS) and fourth ($\tfrac{1}{3}$-TDS) 
row corresponds to the subgraph extracted by Goldberg's exact algorithm \cite{Goldberg84} for the \DSP, 
Charikar's $\tfrac{1}{2}$-approximation algorithm \cite{Char00} for the \DSP, Algorithm 1 
and Algorithm 4 for the \TDSP respectively.
For each optimal extracted subgraph $S$, we show 
its size as a fraction of the total number of vertices, 
the edge density $f_e(S)$, the average degree $\delta(S)=2e(S)/|S|$,
the triangle density $f_t(S)$ and  the average number of triangles per vertex $\tau(S)=3t(S)/ |S|$.
We observe that for all datasets,  the optimal \textsf{\TDS} is close to being a near-clique
while the optimal \textsf{\DS} is not always so. A pronounced example is the \textsf{Football}
network where the optimal \textsf{\DS} is the whole network with $f_e=0.0094$, whereas the optimal \textsf{\TDS}
is a set of 18 vertices with edge density 0.48.
Finally, we observe that the quality of Algorithm's~4 output is very close to the optimal
solution and sometimes even better. It is worth mentioning that the same phenomenon 
is observed in the case of Charikar's $\frac{1}{2}$-approximation algorithm \cite{Char00}
compared to Goldberg's exact algorithm \cite{Goldberg84}.

We use the scalable \textsc{Java} implementation of Algorithm~4
and a scalable implementation of Charikar's $\frac{1}{2}$-approximation algorithm 
on the rest of the datasets of Table~\ref{tab:datasets}.
The results are shown in Table~\ref{tab:morereal}. 
Again, we verify the fact that the \TDSP results in near-cliques,
even when the \DSP fails. 
For instance,  for the collaboration network \textsf{ca-Astro} 
the \DSP results  in a subgraph with 1\,184 vertices with $f_e=0.05,f_t=0.002$.
The \TDSP results in a clique with 57 vertices.
The experimental results in Tables~\ref{tab:motivational} and~\ref{tab:morereal}
strongly indicate that the algorithms developed in this work consitute graph mining 
primitives that can be used to extract near-cliques when the \DSP problem fails to do so.

\begin{table}[t]
\begin{center}
\small
\begin{tabular}{r|ccc|ccc|}
\multicolumn{1}{c}{} &  \multicolumn{3}{c}{$\frac{1}{2}$-DS} &  \multicolumn{3}{c}{$\frac{1}{3}$-TDS} \\
\cline{2-7}
                               &
\multicolumn{1}{c}{$|S|$}      &
\multicolumn{1}{c}{$f_e$}      &
\multicolumn{1}{c|}{$f_t$}     &
\multicolumn{1}{c}{$|S|$}      &
\multicolumn{1}{r}{$f_e$}      &
\multicolumn{1}{c|}{$f_t$}              \\ 
\cline{2-7}  
\textsf{AS-735}           & 59          &  0.28     &  0.08   & 13 & 0.8  &  0.66  \\ 
\textsf{AS-caida}         & 143         & 0.14      &  0.02   & 27 & 0.52  & 0.25 \\ 
\textsf{ca-Astro}         & 1\,184      & 0.05      &  0.002  & 57 & 1     & 1 \\ 
\textsf{ca-GrQC}          & 42          & 0.79      &  0.68   & 14 & 0.89  & 0.84 \\ 
\textsf{Epinions}         & 718         & 0.27      &  0.10   & 135& 0.60  & 0.33 \\ 
\textsf{Enron}            & 192         & 0.30      &  0.07   &139 & 0.40  & 0.12 \\ 
\textsf{EuAll}            & 248         & 0.20      &  0.01   &108 & 0.40  & 0.18 \\ 
\textsf{soc-Slashdot0811} & 1\,637      & 0.29      &  0.08   &253 & 0.52  & 0.29 \\ 
\textsf{soc-Slashdot0902} & 1\,787      & 0.28      &  0.07   &247 & 0.49  & 0.23 \\ 
\textsf{wb-cs-Stanford}   &  84         & 0.64      &  0.48   & 26 & 0.80  & 0.67 \\ 
\cline{2-7}    
\end{tabular}
\caption{\label{tab:morereal}Comparison of the extracted subgraphs by 
the $\frac{1}{2}$-approximation algorithm of Charikar and 
the $\frac{1}{3}$-approximation algorithm, Algorithm 4.}
\end{center}
\end{table}

\subsection{Exploring parameter $\epsilon$ in Algorithm~5}
\label{subsec:epsilon} 

In this Section we present the results of Algorithm~5 on the DBLP graph.
This is particularly interesting instance as it indicates that 
instead of thinking for to select a good $\epsilon$ value,
it is worth trying out at least few  when resources are available. 
We range $\epsilon$ from 0.1 to 1.8 with a step of 0.1. 	
Figure~\ref{fig:peeling}(a) plots the number of rounds 
Algorithm~5 takes to terminate as a function of $\epsilon$. 
We observe that even for small $\epsilon$  values the number of rounds
is 6. The reader should compare this to the upper bound predicted by 
Lemma~\ref{lem:peeloff} which exceeds 100.
Figure~\ref{fig:peeling}(b) plots the relative ratio $Rel.~\tau=\frac{\tau(S)}{\tau_G^*}$
where $S$ is the output of Algorithm~5. For convenience, the lower bound $\frac{1}{3+3\epsilon}$ 
is plotted with red color.
Similarly, Figure~\ref{fig:peeling}(c) plots the relative ratios $\frac{f_e(S)}{f_e(S^*)},\frac{f_t(S)}{f_t(S^*)}$
as a function of $\epsilon$.  As we observe, the quality of Algorithm~5 is close to the optimal
solution except for $\epsilon=0.7$ and $\epsilon=0.8$. By inspecting why this happens 
we observe that the optimal \textsf{\TDS} is a clique of 44 vertices. 
It turns out that for   $\epsilon=0.7,0.8$ the optimal subgraph which is found
in the last round of the execution of the algorithm (the latter happens for all $\epsilon$ values)
consists of 98 and 74 vertices which contain as a subgraph the optimal $K_{44}$. 
For other values of $\epsilon$, the subgraph in the last round is either
the optimal $K_{44}$ or close to it, with few more extra vertices.
This example shows the potential danger of using  a single value for $\epsilon$,
suggesting that trying out a small number of $\epsilon$ values can be 
significantly beneficial in terms of the approximation quality.

\begin{figure*}[htp]
\centering
\begin{tabular}{@{}c@{}@{\ }c@{}@{\ }c@{}}
\includegraphics[width=0.33\textwidth]{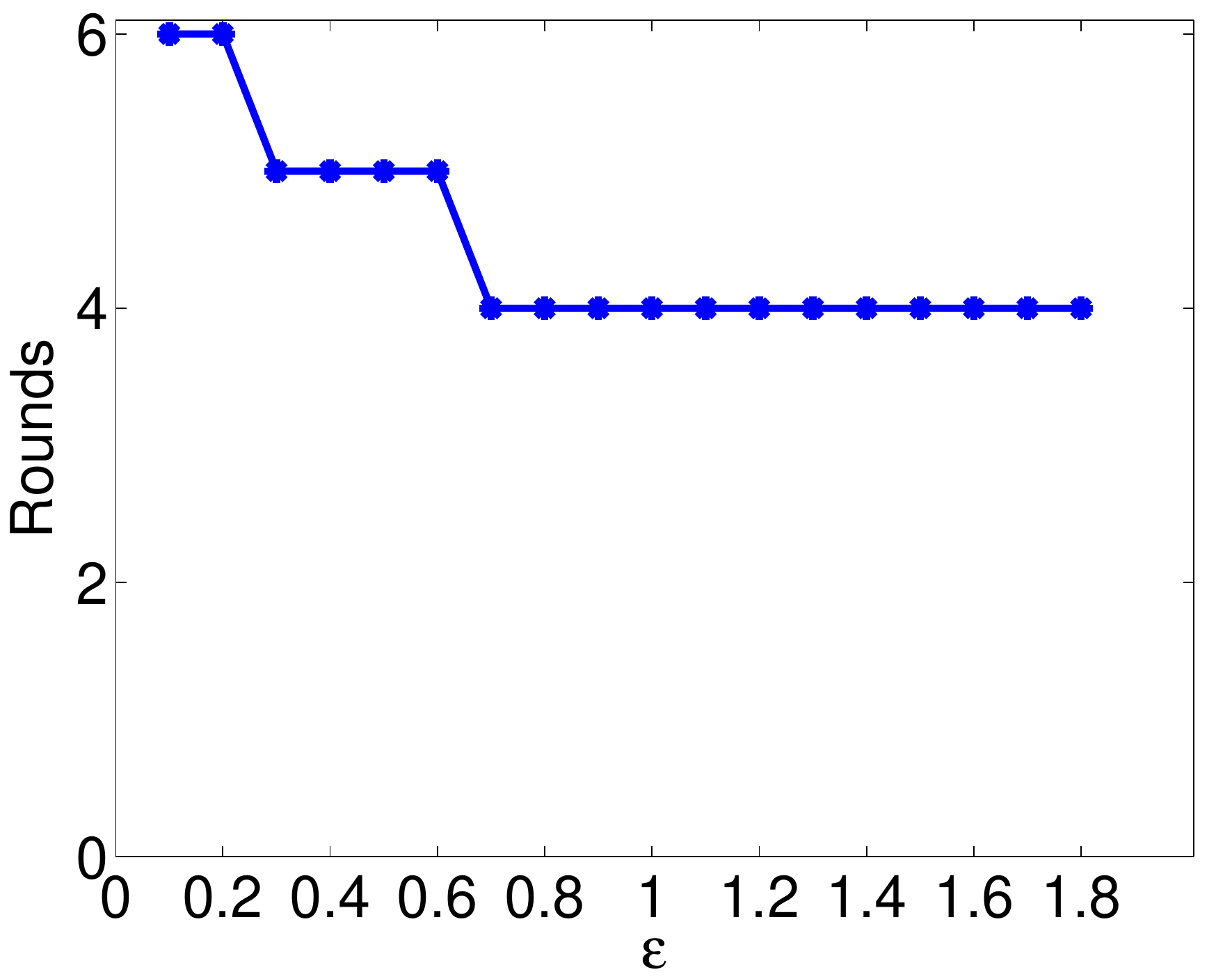} & \includegraphics[width=0.33\textwidth]{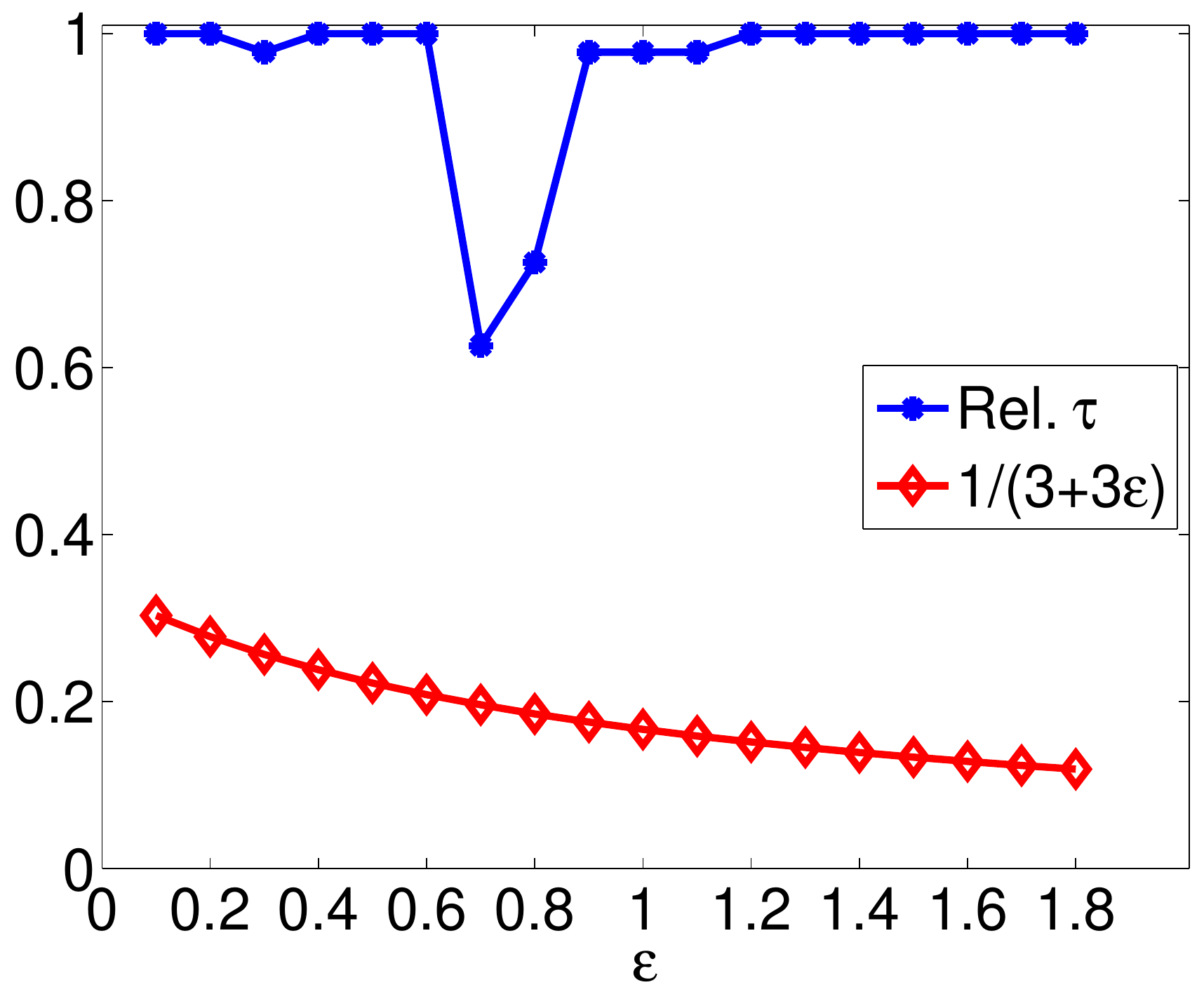}    & \includegraphics[width=0.33\textwidth]{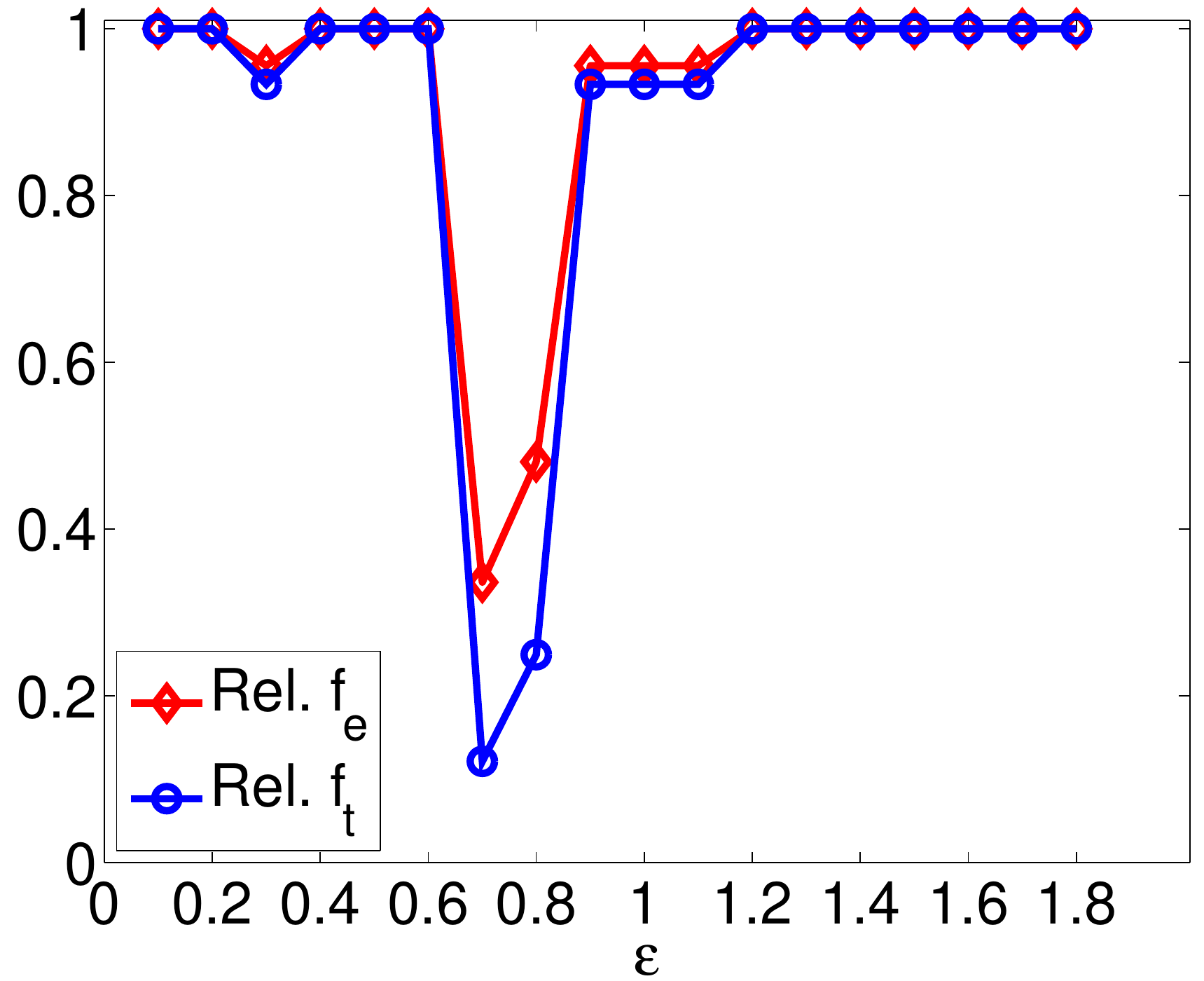}       \\
(a) & (b) & (c)
\end{tabular}
\caption{\label{fig:peeling} Exploring the trade-off between the number of rounds
and accuracy as a function of the parameter $\epsilon$ for Algorithm~5. 
Let $S,S^*$ be the extracted subgraphs by Algorithms~5 and 1 respectively.
(a) Number of rounds, (b) relative average triangle density ratio 
$\frac{\tau(S)}{\tau_G^*}$ (blue $*$) 
and the approximation guarantee $1/(3+3\epsilon)$ (red $ \diamond $), and (c) relative ratios 
$\frac{f_e(S)}{f_e(S^*)},\frac{f_t(S)}{f_t(S^*)}$ as functions of $\epsilon$. 
}
\end{figure*}

\section{Application: Organizing Cocktail Parties}
\label{sec:application}
A graph mining problem that comes up in various applications 
is the following: given a set of vertices $Q \subseteq V$, find 
a dense subgraph containing $Q$
We refer to this type of graph mining problems as {\it cocktail problems}, 
due to the following motivation, c.f. \cite{sozio2010community}. 
Suppose that a set of people $Q$ wants to organize a cocktail party.
How do they invite other people to the party so that the 
set of all the participants, including $Q$, are as similar as possible? 
A variation of the \TDSP which addresses this graph mining problem 
follows.  

\begin{problem}[\CTDSP]
Given $G(V,E)$ and $Q \subseteq V$, find the subset of vertices $S^*$ that maximizes 
the triangle density such that $Q \subseteq S^*$ ,

$$S^*=\arg\max_{Q \subseteq S \subseteq V} \tau(S).$$
\end{problem}

The \CTDSP can be solved by modifying our proposed algorithms 
accordingly. A useful corollary follows. 

\begin{corollary} 
The \CTDSP is solvable in polynomial time by adding arcs from $s$ to $v \in A$
of large enough capacities, e.g., capacities equal to $n^3+1$ are sufficiently large. 
Furthermore, the peeling algorithm which avoids removing vertices from $Q$
is a $\frac{1}{3}$-approximation algorithm for the \CTDSP. 
\end{corollary} 


In the following we evaluate the $\frac{1}{3}$-approximation 
algorithm on two datasets. The two experiments
indicate two different types of performances that should be expected 
in real-world applications. The first is a positive whereas the second 
is negative case. 
Both experiments here serve as sanity checks\footnote{ 
For instance,  by preprocessing the political vote data 
from a matrix form to a graph using a threshold for edge additions, 
results in  information loss. }

\spara{Political vote data.} 
We obtain Senate data  for the first session (2006) of the 109th congress which 
spanned the period from January 3, 2005 to January 3, 2007, 
during the fifth and sixth years of George W. Bush's presidency \cite{wikipedia}. 
In this Congress, there were 55, 45 and 1 Republican, Democratic and independent
senators respectively. The dataset can be downloaded from 
the US Senate web page \url{http://www.senate.gov}. 
We preprocess the dataset in the following way: we add an edge between two senators
if amonge the bills for which they both casted a vote, they voted at least 80\% 
of the times in the same way. The resulting graph has 100 vertices and 
2034 edges. We run the $\frac{1}{3}$-approximation algorithm 
on this graph using as our set $Q$ the first three republicans
according to lexicographic order: Alexander (R-TN), Allard (R-CO) and Allen (R-VA). 
We obtain at our output a subgraph consisting of 47 vertices. By inspecting
their party, we find that 100\% of them are Republicans. 
This shows that our algorithm in this case succeeds in finding the large majority 
of the cluster of republicans. It is interesting that the 8 remaining Republicans
do not enter the \textsf{\TDS}.  A careful inspection of the data, c.f. \cite{press},
indicates that 6 republicans agree with the party vote on at most 79\% of the bills, 
and 8 of them on at most 85\% of the bills.

\spara{DBLP graph.} We input as a query set $Q$ a set of scientists
who have established themselves in theory and algorithm design:
Richard Karp, Christos Papadimitriou, Mihalis Yannakakis and Santosh Vempala. 
The algorithm returns at its output the query set and a set $S$
of 44 vertices corresponding to a clique of (mostly) Italian computer 
scientists. We list a subset of the 44 vertices here: 
M. Bencivenni, M. Canaparo,  F. Capannini, L. Carota, M. Carpene,
R. Veraldi, P. Veronesi, M. Vistoli, R. Zappi. 
The output graph induced by $S \cup Q$ is disconnected. 
Therefore, this can be easily explained because of the following (folklore) inequality,
given that $|Q|<|S|$ in our example.

\begin{claim}
Let $a,b,c,d$ be non-negative. Then,
\begin{equation}
 \max{\big( \frac{a}{c}, \frac{b}{d} \big)} \geq \frac{a+b}{c+d} \geq \min{\big( \frac{a}{c}, \frac{b}{d} \big)}
\end{equation}
\end{claim}

\noindent In our example, we get $a=t(S),c=|S|,b=t(Q),d=|Q|$.
In such a scenario, where the output consists of the union of a 
dense subgraph and the query set $Q$,  
an algorithm which builds itself up from $Q$ -assuming $Q$ is not an independent set- 
to $V$ by adding vertices which create as many triangles as possible and returning 
the maximum density subgraph, 
rather than peeling vertices from $V$ downto $Q$ should be preferred in practice, 
see also \cite{tsourakakis2013denser}.

\hide{
This is the output of the algorithm 

M. Bencivenni, M. Canaparo,  F. Capannini, L. Carota, M. Carpene \\ 
A. Cavalli, A. Ceccanti, M. Cecchi, D. Cesini, A. Chierici \\ 
V. Ciaschini, A. Cristofori Luca dell'Agnello, D. De Girolamo, M. Donatelli \\ 
D. Dongiovanni, E. Fattibene, T. Ferrari, A. Ferraro, A. Ghiselli \\
D. Gregori, A. Italiano, L.  Magnoni, B. Martelli, \\
M. Mazzucato, M. Onofri,  A. Paolini, A. Prosperini \\
P. Ricci, E. Ronchieri, F. Rosso, D. Salomoni, V. Venturi \\ 
R. Veraldi, P. Veronesi, M Vistoli, D. Vitlacil, Riccardo Zappi \\
S. Dal Pra,  A. Forti, G. Guizzunti, G. Misurelli,Giuseppe Misurelli
V. Sapunenko, S. Zani   \\
{\it R. Karp, C. Papadimitriou, P. Raghavan, M. Yannakakis, M. Vardi, S. Vempala} 

}

\section{Conclusion}
\label{sec:conclusion}
In this work we introduce the average triangle density as a novel objective 
for attacking the important problem of finding near-cliques. We propose
exact and approximation algorithms and an efficient \mrc implementation.
Furthermore, we show how to generalize our results to maximizing the 
average $k$-clique density. Experimentally we verify
the value of the \TDSP as a novel addition to the graph mining toolbox. 
Also, we show how to solve a constrained version of the \TDSP which has 
various graph mining applications. 

Our work leaves numerous problems open, including the following:  
(a) Can we obtain a better exact solution? 
(b) How do approximate triangle counting methods affect the outcome of the 
$\frac{1}{3}$-approximation algorithm?
(c) Are there real-world networks where \TDSP fails to extract near-cliques?
In those networks, can the \KCDSP problem for $k$ constant render the situation
in an analogy of how the \TDSP succeeds in cases where the \DSP fails? 
We have implemented the exact algorithm in Section~\ref{subsec:kclique} and we have tested for $k=4$, 
namely maximizing the average $K4$ density.
Preliminary results on various graph datasets suggest that there can be gains when 
one uses 
higher $k$-values but the gain obtained from moving from the \DSP to the \TDSP 
is typically significantly larger from the gain obtained (if any) from the \TDSP to larger $k$ values. 
(d) It is clear that one can extract the top-$k$ non-overlapping triangle densest
subgraphs, using $k$ executions of one of our algorithms. 
For instance, extracting the top-7 \textsf{\TDS}s from the DBLP graph results 
in finding 7 cliques of size 44, 27, 25, 24, 20, 20 and 19\footnote{
The corresponding top-7 results for the \DSP reported as $(|S|,f_e(S))$ are 
(44,1), (27,1), (25,1), (25,1), (64, 0.31), (36,0.48), (89,0.19).}
Can we compute such dense subgraphs simultaneously? 

\hide{ top7 dblp
44	1	44	1
27	1	27	1
25	1	25	1
24	1	24	1
64	0.310515873	20	1
36	0.485714286	20	1
89	0.187946885	19	1
}

\hide{Concerning (a): An approach that does not work in terms 
of giving the proper guarantees, but is not ``far'' off is the following:
 Network $H$ is a weighted and directed version of $G$. 
Specifically, the vertex set of $H$ is $V(H) = \{s\} \cup V(G) \{t\}$. 
Furthermore for each edge $e=\{ u,v\} \in E(G)$ we add two directed 
arcs $(u,v), (v,u)$ in $E(H)$, where the weight of 
each arc $e \in E(G)$ is equal to the number of triangles 
edge $e$ participates in.} 

\section*{Acknowledgements}
I would like to thank Clifford Stein for pointing out 
that we may use \cite{ahuja1994improved} in the place
of other max flow algorithms to obtain a faster algorithm. 
Also, I would to thank Chen Avin, Kyle Fox, Ioannis Koutis, Danupon 
Nanongkai and Eli Upfal for their feedback.

\bibliographystyle{alpha}
\bibliography{ref}

\end{document}